\documentclass[conf]{new-aiaa}

\usepackage{graphicx} % Required for inserting images
\usepackage[dvipsnames]{xcolor}
\usepackage{soul}
\usepackage{xspace}
\usepackage{enumitem}
\usepackage{subcaption}
\usepackage{todonotes}

\usepackage{amsthm}

\allowdisplaybreaks

\usepackage[utf8]{inputenc}
\usepackage{amsmath}
\usepackage[version=4]{mhchem}
\usepackage{siunitx}
\usepackage{longtable,tabularx}
\usepackage{bm}

\newcommand{\ellone}{$\mathcal{L}_1$\xspace}
\newcommand{\elloneDRAC}{$\mathcal{L}_1$-DRAC\xspace}

\newcommand{\Borel}[1]{\mathcal{B} \left( #1 \right)}
\newcommand{\br}[1]{\left( #1 \right)} % Height adjusted round brackets
\newcommand{\sbr}[1]{\left[ #1 \right]} % Height adjusted square brackets
\newcommand{\cbr}[1]{\left\{ #1 \right\}} % Height adjusted curly brackets
\newcommand{\diracmeasure}[2]{\delta_{#1}\left( #2 \right)} % Dirac measure

\newcommand{\Fmu}[1]{F_\mu \left( #1 \right)}
\newcommand{\Fsigma}[1]{F_\sigma \left( #1 \right)}

\newcommand{\Fbarmu}[1]{\bar{F}_\mu \left( #1 \right)}

\newcommand{\ito}{It\^{o}\xspace}

\newcommand{\Lmu}[1]{\Lambda_\mu \left( #1 \right)}

\newcommand{\Lsigma}[1]{\Lambda_\sigma \left( #1 \right)}

\newcommand{\fbar}[1]{\bar{f} \left( #1 \right)}

\newcommand{\fdbkIL}{\pi_{\text{\tiny IL}}}
\newcommand{\fdbktasil}{\pi_{\text{\tiny TaSIL}}}
\newcommand{\fdbkellone}{\pi_{\mathcal{L}_1}}
\newcommand{\ellonedrac}{$\mathcal{L}_1$-DRAC\xspace}

\newcommand{\Wt}[1]{W_{#1}}
\newcommand{\Wfilt}[1]{\mathfrak{W}_{#1}}

\newcommand{\xstart}[1]{x^\star_{#1}}
\newcommand{\xpt}[1]{x'_{#1}}
\newcommand{\xnomt}[2]{x_{#1}\left( #2 \right)}

\newcommand{\pistar}[1]{\pi^\star \left( #1 \right)}
\newcommand{\piIL}[1]{\pi_{\text{\tiny IL}}\left( #1 \right)}
\newcommand{\Xt}[1]{X_{#1}}

\newcommand{\Xdist}[1]{\mathbb{X}_{#1}}

\newcommand{\Wstart}[1]{W^\star_{#1}}

\DeclareFontFamily{U}{stix2bb}{\skewchar\font127 }
\DeclareFontShape{U}{stix2bb}{m}{n} {<-> stix2-mathbb}{}
\DeclareMathAlphabet{\lowermathbb}{U}{stix2bb}{m}{n}

\newcommand{\xstardist}[1]{\lowermathbb{x}^\star_{#1}}
\newcommand{\xpdist}[1]{\lowermathbb{x}'_{#1}}

\newcommand{\Probability}[1]{\mathbb{P}\left( #1 \right)}
\newcommand{\partition}[2]{\mathfrak{p}_{#1}(#2)}

\newcommand{\Yt}[1]{Y_{#1}}
\newcommand{\Wrt}[1]{\widehat{W}_{#1}}
\newcommand{\Gmu}[1]{G_\mu \left( #1 \right)}
\newcommand{\Gsigma}[1]{G_\sigma \left( #1 \right)}
\newcommand{\Ut}[1]{U_{#1}}

\newcommand{\norm}[1]{\left\| #1 \right\|}
\newcommand{\Frobenius}[1]{\norm{#1}_F}

\newcommand{\Boldomega}{ \omega  }

\newcommand{\BoldTs}{ T_s  }

\DeclareMathOperator*{\argmin}{argmin}

\usepackage{scalerel} % For a prettier/smaller parallel 
\newcommand*{\paral}{\stretchrel*{\parallel}{\perp}}
\newcommand{\expo}[1]{e^{ #1 }}
\newcommand{\Lhat}[1]{\hat{\Lambda}\left( #1\right)}
\newcommand{\Lparahat}[1]{\hat{\Lambda}^{\paral}\left( #1\right)}

\usepackage{bbm} % For \mathsc (see notatation of Wasserstein space)
\newcommand{\indicator}[2]{\mathbbm{1}_{#1}\left( #2 \right)} % Indicator function

\NewDocumentCommand{\FL}{o o}{ 
  \pi_{\mathcal{L}_1} \IfValueT{#1}{\left(#1\right)} \IfValueT{#2}{\left(#2\right)}
}

\NewDocumentCommand{\Filter}{o o}{
  \mathcal{F}_\Boldomega \IfValueT{#1}{\left( #1 \right)} \IfValueT{#2}{(#2)} 
} 

\NewDocumentCommand{\FilterW}{o o}{
  \mathcal{F}_{\mathcal{N},\Boldomega} \IfValueT{#1}{\left( #1 \right)} \IfValueT{#2}{(#2)} 
}

\NewDocumentCommand{\AdaptationLaw}{o o o}{
  \mathcal{F}_{\BoldTs} \IfValueT{#1}{\left( #1, \IfValueT{#2}{#2} \right)} \IfValueT{#3}{(#3)} 
}
\NewDocumentCommand{\AdaptationLawParal}{o o}{
  \mathcal{F}^{\paral}_{\BoldTs} \IfValueT{#1}{\left( #1 \right)} \IfValueT{#2}{(#2)} 
}

\NewDocumentCommand{\Predictor}{o o}{
  \mathcal{F}_{\lambda_s} \IfValueT{#1}{\left( #1 \right)} \IfValueT{#2}{(#2)} 
}

\NewDocumentCommand{\ReferenceInput}{o o}{
  \mathcal{F}_r \IfValueT{#1}{\left( #1 \right)} \IfValueT{#2}{(#2)} 
}

\NewDocumentCommand{\ControlError}{o}{
  \widehat{\mathcal{F}} \IfValueT{#1}{\left( #1 \right)}  
}

\newtheorem{definition}{Definition}
\newtheorem{remark}{Remark}[section]
\newtheorem{assumption}{Assumption}
\newtheorem{proposition}{Proposition}[section]
\newtheorem{theorem}{Theorem}[section]

\newtheorem{corollary}{Corollary}[section]

\hypersetup{
    colorlinks=true, 
    linkcolor=red,
    urlcolor=magenta,
    citecolor=blue,
}

\title{Distributionally Robust Imitation Learning: Layered Control Architecture for Certifiable Autonomy}

\author{Aditya Gahlawat\footnote{Research Engineer, Department of Mechanical Science and Engineering, University of Illinois Urbana-Champaign, Urbana, IL 61801, USA}, Ahmed Aboudonia \footnote{Postdoctoral Research Associate, Department of Mechanical Science and Engineering, University of Illinois Urbana-Champaign, Urbana, IL 61801, USA}, Sandeep Banik\footnote{Postdoctoral Research Associate, Department of Mechanical Science and Engineering, University of Illinois Urbana-Champaign, Urbana, IL 61801, USA} and Naira Hovakimyan\footnote{W. Grafton and Lillian B. Wilkins Professor, Department of Mechanical Science and Engineering, University of Illinois Urbana-Champaign, Urbana, IL 61801, USA. AIAA fellow}}
\affil{University of Illinois Urbana-Champaign}
\author{Nikolai Matni\footnote{Assistant Professor, Department of Electrical and Systems Engineering, University of Pennsylvania, Philadelphia, PA 19104, USA.}}
\affil{University of Pennsylvania}
\author{{Aaron D. Ames}\footnote{Bren Professor, Mechanical and Civil Engineering, Control and Dynamical Systems, California Institute of Technology, California 91125, USA}}
\affil{California Institute of Technology}
\author{{Gioele Zardini}\footnote{Rudge (1948) and Nancy Allen Career Development Assistant Professor, Massachusetts Institute of Technology, Laboratory for Information and Decision Systems (LIDS), Department of Civil and Environmental Engineering, Institute for Data Systems and Society (IDSS), Cambridge, MA 02139-4307, USA}}
\affil{Massachusetts Institute of Technology}
\author{{Alberto Speranzon}\footnote{Chief Scientist, Advanced Technology Labs, Lockheed Martin, 1303 Corporate Center Drive, Eagan, Minnesota, MN 55121, USA.}}
\affil{Advanced Technology Labs, Lockheed Martin}

\begin{document}

% \listoftodos

% %%%%%%%%%%%%%%%%%%
% \newpage

\maketitle

\begin{abstract}
Imitation learning (IL) enables autonomous behavior by learning from expert demonstrations. 
While more sample-efficient than comparative alternatives like reinforcement learning, IL is sensitive to compounding errors induced by distribution shifts. 
There are two significant sources of distribution shifts when using IL-based feedback laws on systems: distribution shifts caused by policy error and distribution shifts due to exogenous disturbances and endogenous model errors due to lack of learning. 
Our previously developed approaches, Taylor Series Imitation Learning (TaSIL) and $\mathcal{L}_1$ -Distributionally Robust Adaptive Control (\ellonedrac), address the challenge of distribution shifts in complementary ways. 
While TaSIL offers robustness against policy error-induced distribution shifts, \ellonedrac offers robustness against distribution shifts due to aleatoric and epistemic uncertainties. 
To enable certifiable IL for learned and/or uncertain dynamical systems, we formulate \textit{Distributionally Robust Imitation Policy (DRIP)} architecture, a Layered Control Architecture (LCA) that integrates TaSIL and~\ellonedrac. 
By judiciously designing individual layer-centric input and output requirements, we show how we can guarantee certificates for the entire control pipeline. 
Our solution paves the path for designing fully certifiable autonomy pipelines, by integrating learning-based components, such as perception, with certifiable model-based decision-making through the proposed LCA approach.
% Our solution opens up further forays into the design of certifiable complete autonomy pipelines by, e.g., stacking perception and decision-making by following the presented LCA approach.
\end{abstract}

%%%%%%%%%%%%%%%%%%%%%%%%%%%%%%%%%%%%%%%%%%%%%
\section{Introduction}\label{sec:Introduction}
%%%%%%%%%%%%%%%%%%%%%%%%%%%%%%%%%%%%%%%%%%%%%
% Outline 1(a) and 1(b)
Learning from expert demonstrations, particularly Imitation Learning (IL)~\cite{husseinImitationLearningSurvey2018}, has emerged as a powerful framework for synthesizing control policies directly from data. 
Compared to reinforcement learning (RL), IL has shown improved data efficiency enabling its successful deployments in domains ranging from robotics~\cite{ankileJUICERDataEfficientImitation2024,ciftciSAFEGILSAFEtyGuided2024a} and autonomous driving~\cite{pomerleauALVINNAutonomousLand1988,codevillaEndtoEndDrivingConditional2018} to aerial vehicles~\cite{abbeelApprenticeshipLearningInverse2004a}, navigation~\cite{husseinDeepImitationLearning2018}, and gaming~\cite{rossReductionImitationLearning2011}.
%%%%%%%%%%%%%
% Outline 2(a) and 2(b).

Despite this promise, IL inherits a fundamental limitation shared by all learning-based feedback methods: its dependence on training data.
Any deviation between the states encountered during deployment and those present in the expert demonstrations can lead to compounding errors, resulting in degraded performance, and, in safety-critical systems, potential loss of stability.
This mismatch, commonly termed \emph{distribution shift}, appears in IL as the \emph{imitation gap}, the performance difference between the expert and learned policies. 
Because learned policies rarely reproduce expert behavior exactly, this imitation gap accumulates over time, giving rise to \emph{policy-induced} distribution shift.
The inability of standard IL methods to guarantee robustness to such distribution shift remains a central barrier to their adoption in certifiable autonomous systems.
This vulnerability is well documented in the IL literature and has motivated various mitigation strategies, including dataset aggregation (e.g., DAgger~\cite{rossReductionImitationLearning2011a}), hybrid approaches combining IL with reinforcement learning (e.g., on-policy~\cite {sunDualPolicyIteration2018,hoGenerativeAdversarialImitation2016a}), and noise injection~\cite{laskeyDARTNoiseInjection2017}. 
A majority of the approaches, however, rely on unverifiable assumptions, such as access to an interactive expert or repeated system interaction via a simulator.
For safety-critical systems, such assumptions are often prohibitively expensive, if not impossible.      

The term ``imitation gap'' typically suggests that distribution shift arises solely because the learned policy fails to accurately reproduce the expert policy.
However, this view is incomplete: it overlooks both the \emph{quality} of the model used for policy synthesis and the \emph{exogenous disturbances} that inevitably affect real systems.
Any model used for policy synthesis, be it data-driven or obtained via classical system identification approaches using scale models in a wind tunnel, will be inherently imperfect, and thus, subject to \emph{epistemic} (learnable) uncertainties. 
Additionally, the systems are subject to \emph{aleatoric} (unbearable but statistically descriptive) uncertainties.
Finally, ambiguity in the system's initialization state propagates over time and can further exacerbate the effects of epistemic and aleatoric uncertainties.
\ul{Existence of any subset or combination of epistemic, aleatoric, and initialization ambiguity will contribute to a \emph{distribution shift} even if one assumes perfect imitation of the expert policy.}

One may seek to mitigate these uncertainties leveraging robust or adaptive control techniques, but most existing approaches depend on uncertainty models that introduce substantial conservatism~\cite{mitchellTimedependentHamiltonJacobiFormulation2005,herbertFaSTrackModularFramework2017,singhRobustFeedbackMotion2023}.
For example, uncertainty representations that are amenable for synthesis but poor in representational accuracy include bounded uncertainties assumed to belong to known compact sets, parametrized uncertainties with known basis functions, and parametric distributions such as Gaussian.
In addition to being surrogates that generalize poorly, such representations fail to capture uncertainties in data-driven learned systems, whose properties are guaranteed only for the empirical distributions supported by the training data.
We therefore argue that the safe and predictable deployment of IL policies requires controllers that are robust to \textbf{both policy-induced and uncertainty-induced distribution shifts}, while simultaneously providing \emph{a priory} certifiable guarantees on their performance.
%We thus argue that a reasonable requirement for safe and predictable deployment of IL policies is that the synthesized policies guarantee robustness to the effects of \textbf{both policy and uncertainty-induced distribution shifts while also providing \emph{a priori} guarantees for certifiability.}    

\subsection{Prior art}
We now discuss the existing results in the literature that provide results for problems similar to the one we consider in this manuscript.

\noindent \textbf{Imitation Learning}:
Behavioral cloning (BC) represents the simplest instantiation of IL, formulated as supervised regression over state--action pairs, with early demonstrations in autonomous driving and navigation~\cite{pomerleauALVINNAutonomousLand1988, panAgileAutonomousDriving2019}.
In parallel, IL was recognized as a principled mechanism for motor skill acquisition, emphasizing that complex behaviors can emerge through the reuse and generalization of demonstrated trajectories~\cite{schaalImitationLearningRoute}.
These ideas influenced learning from demonstration and programming by demonstration paradigms in robotics, where policies are obtained by mapping observed states to expert actions or motion primitives~\cite{calinonRobotProgrammingDemonstration2009, argallSurveyRobotLearning2009}.
Such approaches emphasized simplicity, data efficiency, and compatibility with high-dimensional perception.

Subsequent work broadened the scope of IL to address limitations of BC and to improve generalization. 
Inverse reinforcement learning (IRL) recasts IL as the problem of inferring an underlying reward function that yields expert behavior, enabling policy optimization under the learned reward function\cite{abbeelApprenticeshipLearningInverse2004c, abbeelApprenticeshipLearningInverse2004a}. 
Adversarial imitation learning \cite{hoGenerativeAdversarialImitation2016a} formulations further refine this idea by matching expert and learner distribution of state-action pairs without explicitly recovering reward function.
Numerous surveys have unified and organized IL methods along axes of supervision, representation, and optimization, and documenting their application across manipulation, locomotion, and autonomous driving \cite{husseinImitationLearningSurvey2018, ROB-053, ravichandarRecentAdvancesRobot2020, lemeroSurveyImitationLearning2022}.
 % Collectively, the prior stated set of works establishes IL as a versatile and effective paradigm for learning complex behaviors from data.
 % However, it also exposes a fundamental vulnerability, that the policies trained on expert-induced data distributions can fail under closed-loop execution when there is a mismatch between the training and operating distributions.
 % This observation motivates the next line of prior art on imitation learning methods designed to be robust to policy-induced distribution shifts.

\noindent \textbf{IL Robust to Policy Shifts}:
Interactive IL methods mitigate policy-induced distribution shift by explicitly correcting the mismatch between training and deployment distributions.
DAgger iteratively augments the dataset with expert labels on states visited by the learned policy, thereby aligning training and execution distributions~\cite{rossReductionImitationLearning2011}.
Subsequent variants reduce expert burden or target critical states, including human-gated intervention~\cite{kellyHGDAggerInteractiveImitation2019}, budget-aware querying~\cite{hoqueThriftyDAggerBudgetAwareNovelty2021}, and query-efficient schemes for perception-driven tasks~\cite{zhangQueryEfficientImitationLearning2016}.
Complementary approaches inject stochasticity or adversarial perturbations into demonstrations to broaden coverage of the expert support, as in DART \cite{laskeyDARTNoiseInjection2017}, or estimate the support of expert behavior directly to discourage out-of-distribution actions \cite{wangRandomExpertDistillation2019}. 
Other methods recast robustness as a regularization problem, penalizing disagreement between trained and deployed policies \cite{Brantley2020Disagreement-Regularized} or blending IL with sparse-reward reinforcement learning to recover from distribution shifts \cite{reddySQILImitationLearning2019}.

More recent work moves beyond data aggregation to relax assumptions of persistent expert access or high-fidelity simulators.
Active and coaching-based methods reinterpret expert interaction through the lens of active learning or corrective feedback, selectively requesting supervision when the learned policy deviates from desired behavior \cite{judahActiveImitationLearning2012, heImitationLearningCoaching2012, leSmoothImitationLearning2016}. 
Trajectory-level abstractions further mitigate compounding error by reducing sensitivity to fine-grained action deviations, such as waypoint-based imitation in long-horizon manipulation tasks \cite{shiWaypointBasedImitationLearning2023}. 
A particularly notable advance is Taylor Series Imitation Learning (TaSIL), which explicitly models the effect of policy errors on future states via a local Taylor expansion of the dynamics \cite{NEURIPS2022_7f10c3d6}.
TaSIL incorporates higher-order sensitivity information into the learning objective, directly penalizing error directions that amplify under closed-loop dynamics. 
% This formulation provides a principled bridge between imitation learning and control-theoretic notions of stability, revealing how robustness to policy-induced distribution shift depends on the interaction between policy approximation error and system dynamics. 

\noindent \textbf{IL for Uncertain Systems and Environments}:

\begin{itemize}
    \item \underline{Data-augmentation/training based approaches}: A significant body of work addresses uncertainty by appending the training data or modifying the learning objective to improve robustness.
    Expert demonstrations may be noisy, suboptixmal, or inconsistent, motivating methods that explicitly model demonstrator quality or leverage failed demonstrations as informative signals~\cite{grollmanRobotLearningFailed2012, wuImitationLearningImperfect2019, sasaki2021behavioral, beliaevImitationLearningEstimating2022}.
    Uncertainty due to perception and embodiment mismatch is addressed through third-person imitation and imitation from observation, which learn invariant representations across viewpoints or embodiments~\cite{stadieThirdPersonImitationLearning2019, liuImitationObservationLearning2018}.
    Self-supervised representation learning, including time-contrastive objectives, further enables imitation from raw sensory data in unstructured environments~\cite{sermanetTimeContrastiveNetworksSelfSupervised2018}.
    Other approaches explicitly model distributions over behaviors or trajectories, using probabilistic policies to handle multimodality~\cite{caiProbabilisticEndtoEndVehicle2020} or encouraging robustness across diverse expert behaviors and environment variations~\cite{wangRobustImitationDiverse2017, chaeRobustImitationLearning2022}.
    Hybrid methods that combine IL with reinforcement learning, such as example-guided policy optimization, use demonstrations as anchors while adapting to stochastic dynamics and contact-rich interactions~\cite{pengDeepMimicExampleguidedDeep2018}.
    Recent work on stability-aware density modeling and robust online learning from humans further constrains policy adaptation under uncertainty and human feedback~\cite{kangLyapunovDensityModels2022, mehtaStROLStabilizedRobust2024}.
    \item \underline{System-theoretic approach}: A complementary line of work integrates IL with control-theoretic tools to provide formal guarantees despite modeling uncertainty.
    Methods combining IL with model predictive control or Lyapunov-based analysis establish stability under bounded model mismatch and disturbances~\cite{hertneckLearningApproximateModel2018, leeSafeEndtoendImitation2019, yinImitationLearningStability2022}.
    These guarantees have been further formalized through finite-sample analyses that characterize the data requirements needed to ensure stability under uncertain dynamics~\cite{tuSampleComplexityStability2022}, as well as convergence results that isolate the role of system structure in linear--quadratic settings~\cite{caiGlobalConvergenceImitation2019}.
    While these approaches offer stronger theoretical assurances, they often rely on restrictive modeling assumptions or limited system classes.
\end{itemize}

\noindent Collectively, these works demonstrate that existing IL methods for uncertain systems either rely on strong assumptions, provide limited or system specific guarantees, or lack certifiable performance, motivating approaches that unify robustness, stability, and learning under both policy-induced and uncertainty-induced distribution shifts.

%%%%%%%%%%%%%%%%%%%%%%%%%%%%%%%%%%%%%%%
\subsection{Contributions}

% The training data collected from expert demonstrations originate from the \emph{unknown true system}, whereas IL policies are trained using a nominal or approximate model. 
% These discrepancies imply that the classical notion of the imitation gap is incomplete and requires a more comprehensive definition of the \emph{total imitation gap}, which quantifies the performance loss arising from policy-induced distribution shift together with disturbance-induced distribution shift (aleatoric, epistemic, and initialization uncertainty).

% Taylor Series Imitation Learning (TaSIL)~\cite{NEURIPS2022_7f10c3d6} mitigates policy-induced distribution shift by exploiting the notion of incremental input-to-state stability via a loss augmentation.
% The augmented loss is characterized by higher-order terms of the Taylor series expansion of the learned and expert policies.
% TaSIL achieves generalization bounds which scale with number of expert demonstration without requiring access to interactive expert or a simulator reproducing the true system behavior.

We propose a decoupled approach to achieve \textbf{guaranteed robustness against both policy and uncertainty-induced distribution shifts.}
Our approach brings together our previously developed approaches TaSIL~\cite{NEURIPS2022_7f10c3d6}, and \ellone-distributionally robust adaptive control (\elloneDRAC)~\cite{L1DRAC} within a layered control architecture (LCA)~\cite{matniTheoryDynamicsControl2016, matniTheoryControlArchitecture2024, matniQuantitativeFrameworkLayered2024} framework, termed \textit{\textbf{Distributionally Robust Imitation Policy (DRIP)}}.
The approach of TaSIL offers robustness to policy-induced distribution shift via loss augmentation to learn policies exploiting input-to-state stability properties of the known system. 
On the other hand, the feedback law \elloneDRAC enforces the nominal behavior of TaSIL by guaranteeing robustness against uncertainty-induced distribution shifts. 
Furthermore, \elloneDRAC provides certificates of robustness in the space of probability measures (distributions). 
We unify TaSIL and \elloneDRAC under the layered control architecture (LCA) such that the adage of ``\emph{the whole is greater than the sum of its parts}'' holds true. 
Individually, neither TaSIL, nor \elloneDRAC can offer robustness to the entire spectrum of sources that lead to distribution shifts.  
Instead, by combining the individual components in a layered architecture, we unify the robustness axes of each to certifiably counter the wide spectrum of distribution shift sources. 
The key features of DRIP are as follows:
\begin{itemize}
    \item The low-level control law offers robustness certificates in the form of ambiguity sets in space of Borel measures under the Wasserstein metric. 
    The robust certificates are guaranteed by design and are sample-free;  
    \item The $\mathcal{L}_1$-DRAC controller is based on the architecture of $\mathcal{L}_1$~\cite{hovakimyan2010ℒ1}, which, as demonstrated before in~\cite{lakshmanan2020safe,gahlawat2021contraction,sungrobust}, is well-suited for integration with a high-level controller like TaSIL;
    \item The space of certificates that the unified LCA guarantees opens up further forays into robust use of high-dimensional perception system and other data-driven components that demonstrate high-performance empirically, but lack the robustness guarantees needed for predictable operation. 
    \item \textbf{Train Once, Train TaSIL}: The decoupled approach, along with the training-free nature of \elloneDRAC, allows us to synthesize the DRIP without making any changes to the TaSIL training procedure. That is, our approach achieves robustness to uncertainties via a robust adaptive controller without using data-driven techniques like adversarial training or domain randomization. 
\end{itemize}

%%%%%%%%%%%%%%%%%%%%%%%%%%%%%%%%%%%%%%%
\subsection{Organization}
%%%%%%%%%%%%%%%%%%%%%%%%%%%%%%%%%%%%%%%%%%%%%%%%%%%%%%%%%%%%
We present the system under consideration, followed by the required assumption, the training data for IL, and problem formulation in Section~\ref{sec:proset}. 
In Section~\ref{sec:Methodology}, we present our core results on the combined robustness analysis of TaSIL and \elloneDRAC.
Section~\ref{sec:decouple_distribution_shifts} characterizes the decoupled nature of the policy-induced and uncertainty-induced distribution shifts along with assumption required for the stability of the system.
In Section~\ref{sec:PolicyIG}, we derive the bounds on the imitation gap due to the learned policy, and in Section~\ref{sec:UncertaintyIG}, we derive the bounds as a result of uncertainty, followed by discussion on robust by design in Section~\ref{sec:discussion}. 
We validate our result on numerical experiments in Section~\ref{sec:numerical} and conclude the paper in Section~\ref{sec:conclusion}.
\begin{figure}[h]
    \centering
    \includegraphics[width=0.6\linewidth]{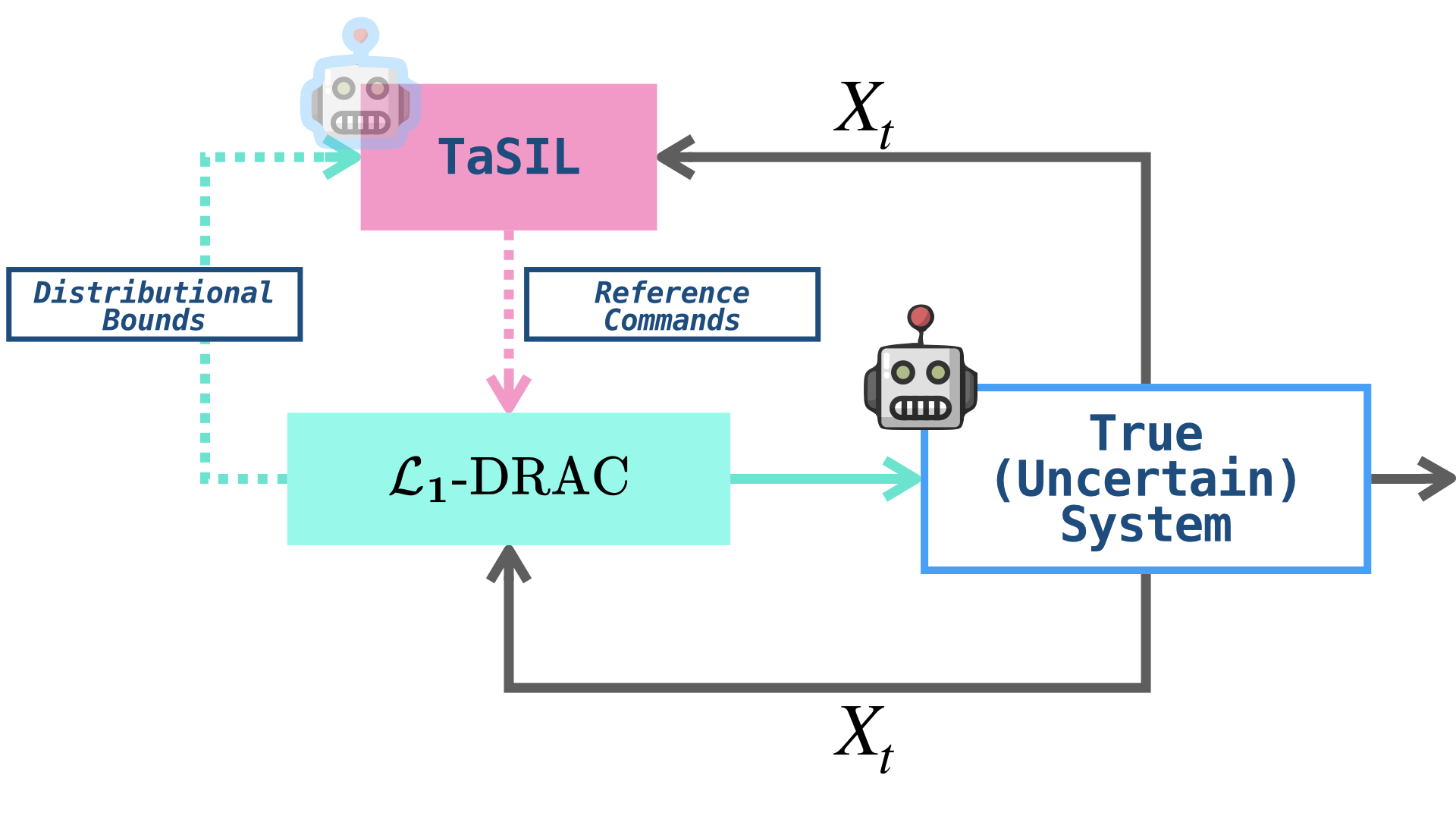}
    \caption{ 
        % \todo[inline]{To be edited to reflect Distributionally Robust Imitation Policy (DRIP)}
        Illustration of a layered control architecture that integrates TaSIL and \elloneDRAC, where $X_{t}$ represents the state of the system. 
        In this architecture TaSIL operates as a mid-level controller that generates reference commands for the low-level \elloneDRAC. 
    }
    \label{fig:LCA_TaSIL_DRAC}
\end{figure}
%%%%%%%%%%%%%%%%%%%%%%%%%%%%%%%%%%%%%%%%%%%%%%%%%%%%%%%%%%%%

%%%%%%%%%%%%%%%%%%%%%%%%%%%%%%%%%%%%%%%
\subsection{Notation}

We denote by $\mathbb{R}_{>0}$ and $\mathbb{R}_{\geq 0}$ the set of positive and non-negative reals, respectively. $\mathbb{N}_0$ is the set of natural numbers starting at $0$. $\mathcal{C}\left(\mathbb{R}^n;\mathbb{R}^m\right)$ denotes the set of continuously differentiable maps $\mathbb{R}^n \to \mathbb{R}^m$. We denote by $\Borel{F}$, the Borel $\sigma$-algebra generated by $F$. 
% $\mathds{1}_{S}$ denotes the indicator set of the set $S$. 
% Furthermore, let $(S,\Sigma, \mathcal{M})$ be a measure space and $1 \leq p < \infty$, $\|f\|_{L_p}$ denotes the $L_p$ norm given by $\|f\|_{L_p} := \left(\int_{S} |f|^p d \mathcal{M} \right)^{1/p}$ and $\pWass{2\sfp}{\Xdist{t}}{\Xstardist{t}}$ denotes the $2$ Wasserstein metric between probability measures $\Xdist{t}$ and $\Xstardist{t}$ and $\mathbb{S}^n$ denotes the set of symmetric matrices $\mathbb{R}^n \to \mathbb{R}^n$. 
We denote $\mathbb{I}_n \in \mathbb{S}^n$ as the identity matrix of dimension $n$, $0_{m,n}, 1_{m,n}$ as the matrices $\mathbb{R}^n \to \mathbb{R}^m$ with all entries equal to $0$ and $1$ respectively.

% {\color{blue}{Notation}
% Borel, Expectation wrt measures, $0_m$
% }

%%%%%%%%%%%%%%%%%%%%%%%%%%%%%%%%%%%%%%%%%%%%%%%%%%%%%%%%%%%%%%%%%%%%%%%%%%%%%%%%

%%%%%%%%%%%%%%%%%%%%%%%%%%%%%%%%%%%%%%%%%%%%%%%%%%%%%%%%%%%%%%%%%%%%%%%%%%%%%%%%
\section{Problem Setup}\label{sec:proset}

In this section we define the systems under consideration, the training data available for policy synthesis, the assumptions on the system dynamics, and the problem statement.

%%%%%%%%%%%%%%%%%%%%%%%%%%%%%%%
\subsection{The Systems}
%%%%%%%%%%%%%%%%%%%%%%%%%%%%%%%
We begin with the definitions of the processes under consideration. 
The following defines the known and unknown drift and diffusion vector fields.
%%%%%%%
\begin{definition}[Vector Fields]\label{def:VectorFields}
    Consider the \textbf{known functions} $f:\mathbb{R}_{\geq 0} \times \mathbb{R}^n \rightarrow \mathbb{R}^n$, and $g: \mathbb{R}_{\geq 0} \rightarrow  \mathbb{R}^{n \times m}$, for $n,m \in \mathbb{N}$.
    Consider also the \textbf{unknown functions} $\Lambda_\mu :\mathbb{R}_{\geq 0} \times \mathbb{R}^n \rightarrow \mathbb{R}^n$ and $\Lambda_\sigma :\mathbb{R}_{\geq 0} \times \mathbb{R}^n \rightarrow \mathbb{R}^{n \times d}$, for $d \in \mathbb{N}$.
    We denote by
    \begin{align}\label{eqn:TrueVectorFields}
            \Fmu{t,a,b} \doteq  f(t,a) + g(t) b + \Lmu{t,a}  \in \mathbb{R}^n, \quad
            \Fsigma{t,a} \doteq   \Lsigma{t,a}  \in \mathbb{R}^{n \times d},
    \end{align}
    for all $\cbr{a,b,t} \in \mathbb{R}^n \times \mathbb{R}^m \times \mathbb{R}_{\geq 0}$, the \textbf{true (uncertain) drift and diffusion vector fields}, respectively.
    Similarly, for any $\cbr{a,b,t} \in \mathbb{R}^n \times \mathbb{R}^m \times \mathbb{R}_{\geq 0}$, we denote by
    \begin{align}\label{eqn:NominalVectorFields}
            \Fbarmu{t,a} \doteq  f(t,a)  \in \mathbb{R}^n,
            \quad 
            \fbar{t,a,b} \doteq \Fbarmu{t,a} + g(t) b \in \mathbb{R}^n,
    \end{align}
    the \textbf{nominal (known) drift vector fields}.
    Note that we can also write
    \begin{subequations}\label{eqn:VectorFields:Decomposition}
        \begin{align}
            \Fmu{t,a,b} 
            = 
            \Fbarmu{t,a} + g(t) b + \Lmu{t,a}
            =
            \fbar{t,a,b} + \Lmu{t,a} 
            \in \mathbb{R}^n,
            \\
            \Fsigma{t,a} = \Lsigma{t,a}  \in \mathbb{R}^{n \times d}.
        \end{align}    
    \end{subequations}
    
\end{definition}
%%%%%%%
Next, we define the processes under consideration in the manuscript.
We follow the notation in~\cite{oksendal2013stochastic} for quantities related to continuous-time stochastic processes.
%%%%%%%%%%%%%%%%% 
\begin{definition}[Systems]\label{def:MainSystems}
    Let $\br{\Omega, \mathcal{F}, \mathbb{P}}$ be a complete probability space, which will be the underlying space throughout the manuscript.
    Let us denote by $\Wt{t}$ and $\Wfilt{t}$ a $\mathbb{P}$-Brownian motion and the filtration it generates, respectively.
    We also define $\Wfilt{\infty} = \sigma \br{ \cup_{t \geq} \Wfilt{t}}$.
    Let $\xi \sim \mathcal{D}$ and $\bar{\xi} \sim \bar{\mathcal{D}}$ be two $\mathbb{R}^n$-valued random variables that are independent of the $\sigma$-algebra $\Wfilt{\infty}$, where $\mathcal{D}$ and $\bar{\mathcal{D}}$ are the respective distributions (probability measures) on the Borel $\sigma$-algebra $\Borel{\mathbb{R}^n}$.
    The probability measure $\mathcal{D}$ is assumed to be supported on a known compact set $\mathcal{X} \subset \mathbb{R}^n$
    Then, we define $\Wfilt{0,t} = \sigma\br{\bar{\mathcal{D}}} \vee \Wfilt{t}$.
    
    Let us denote by the absolutely continuous maps $\pi^\star,~\fdbkIL,~\pi_{\mathcal{L}_1}: \mathbb{R}_{\geq 0} \rightarrow \mathbb{R}^m$, the \textbf{expert}, \textbf{imitation learned}, and \textbf{\ellonedrac feedback processes}, respectively.
    Then, using the known vector fields in Def.~\ref{def:VectorFields}, we denote by 
    \begin{align}\label{eqn:expanded_solutions}
        \xstart{t} \doteq  \xnomt{t}{\xi; \pi^\star} 
        \quad \text{and} \quad 
        \xpt{t} \doteq  \xnomt{t}{\xi; \fdbkIL},
    \end{align}
    the \textbf{expert} and \textbf{nominal (IL) trajectories}, respectively, if they are the solutions to the following ODEs:
    \begin{subequations}\label{eqn:NominalProcesses}
        \begin{align}
            d\xstart{t} = \fbar{t,\xstart{t}, \pistar{\xstart{t}}}, \quad \xstart{0} = \xi \sim \mathcal{D},&
            \quad \text{Nominal (Expert)}
            \label{eqn:ExpertProcess} 
            \\  
            d\xpt{t} = \fbar{t,\xpt{t}, \piIL{\xpt{t}}}, \quad \xpt{0} = \xi \sim \mathcal{D}.&
            \quad \text{Nominal (IL)}
            \label{eqn:TaSILProcess}
        \end{align}
    \end{subequations}
    Similarly, using the unknown vector fields in Def.~\ref{def:VectorFields}, we denote by 
    \begin{equation}
        \Xt{t} \doteq \Xt{t}\left(\bar{\xi},\, \Wt{}; \pi_{ad} \right), \quad
        \pi_{ad} \doteq \fdbkIL + \fdbkellone,
    \end{equation} 
    the \textbf{uncertain (true) trajectory}, if it is the strong solution to the following \ito stochastic differential equation (SDE) adapted to $\Wfilt{0,t}$:
    \begin{align}\label{eqn:L1DRACProcess}
        d\Xt{t} 
        = 
        \Fmu{t,\Xt{t}, \pi_{ad}\br{\Xt{t}} }dt + \Fsigma{t,\Xt{t}}d\Wt{t}, 
        \quad \Xt{0} = \bar{\xi} \sim \bar{\mathcal{D}},
        \quad \text{Uncertain (True)}.
    \end{align}
    We define the \textbf{(instantaneous) law} $\Xdist{t}$ of $\Xt{t}$ as
    \begin{align*}
        \Xdist{t}\br{B} \doteq \Probability{\Xt{t}^{-1}\br{B}}
        = 
        \Probability{\Xt{t} \in B}
        ,
        \quad
        (t,B) \in [0,T] \times \Borel{\mathbb{R}^n},
    \end{align*}
    where $\Xt{t}^{-1}$ denotes the pushforward of $\mathbb{P}$ under $\Xt{t}$.
    We similarly denote by $\xstardist{t}, \xpdist{t}: \Borel{\mathbb{R}^n} \rightarrow [0,1]$ the laws for ~\eqref{eqn:ExpertProcess} and~\eqref{eqn:TaSILProcess}, respectively.
\end{definition}
%%%%%%%%%%%%%%%%% 
A few remarks regarding the systems defined above are in order.
%%%%%%%%%%%%%%%%%
\begin{remark}
    As is standard in imitation learning literature, the nominal systems~\eqref{eqn:NominalProcesses} are free of aleatoric uncertainties. 
    As mentioned in Section~\ref{sec:Introduction}, a subset of existing IL literature considers aleatoric uncertainties as additive statistically determinant disturbances, while others consider epistemic uncertainties. 
    However, the presence of \textbf{both} aleatoric and epistemic uncertainties in the true system~\eqref{eqn:L1DRACProcess} that we consider in this paper is substantially more general as it introduces nonlinear interaction between the epistemic and aleatoric uncertainties.
    For such general systems, beyond the well-posedness in the strong sense, one cannot make any claims on the distributional properties of the process laws.     
\end{remark}
%%%%%%%%%%%%%%%%%

We make the following assumptions on the vector fields for the purposes of well-posedness and analysis.
%%%%%%%%%%%%%%%%%
\begin{assumption}[Regularity and Bounds]\label{assmp:VectorFields}
    The unknown functions $\Lambda_\mu$ and $\Lambda_\sigma$, presented in Definition~\ref{def:VectorFields}, satisfy 
    \begin{align*}
        \norm{\Lmu{t,a}}^2 \leq \Delta_\mu^2 \br{1 + \norm{a}^2}, \quad 
        \norm{\Lsigma{t,a}}_F^2 \leq \Delta_\sigma^2 \br{1 + \norm{a}^2}^\frac{1}{2},~\forall (t,a) \in \mathbb{R}_{\geq 0} \times \mathbb{R}^n, 
    \end{align*}
    where $\Delta_{\mu}, \Delta_{\sigma} \in \mathbb{R}_{>0}$ are known
     The input operator $g: \mathbb{R}_{\geq 0} \rightarrow \mathbb{R}^{n \times m}$ has full column rank $\forall t \in \mathbb{R}_{\geq 0}$, and satisfies 
    \begin{align*}
        g \in \mathcal{C}^1([0,\infty);\mathbb{R}^{n \times m}), \quad \norm{g(t)}_F \leq \Delta_g, \forall t \in \mathbb{R}_{\geq 0},
    \end{align*}
    where $\Delta_g \in \mathbb{R}_{>0}$ is assumed to be known.
    Since $g(t)$ is full rank, we can construct a $g^\perp: \mathbb{R}_{\geq 0} \rightarrow \mathbb{R}^{n \times n-m}$ such that $\text{\emph{Im}} \, g(t)^\perp = \text{\emph{ker}} \, g(t)^\top$ and $\text{\emph{rank}}\br{\bar{g}(t)}=n$, $\forall t \in \mathbb{R}_{\geq 0}$ where 
    \begin{align*}
        \bar{g}(t) \doteq \begin{bmatrix} g(t) & g(t)^\perp  \end{bmatrix} \in \mathbb{R}^{n \times n}.
    \end{align*}
\end{assumption}
%%%%%%%%%%%%%%%%%% 

% %%%%%%%%%%%%%%%%%%%%%%%%%%%%%%%%%%%%%%%%%%%%%%%%%%%%%%%%%%%%
% \begin{figure}[t]
%     \centering
%     \includegraphics[width=0.6\linewidth]{Figures/LearningSetup.jpeg}
%     \caption{The expert trajectories available for imitation learning in our formulation are generated by the uncertain (true) system operating under some expert input process. 
%     The expert trajectory data is then used to learn a nominal model whose predictive performance is only guaranteed on the expert trajectories.}
%     \label{fig:LearningSetup}
% \end{figure}
%%%%%%%%%%%%%%%%%%%%%%%%%%%%%%%%%%%%%%%%%%%%%%%%%%%%%%%%%%%%
%%%%%%%%%%%%%%%%%%%%%%%%%%%%%%%%%%%%%%%%%%%%%%%%%%%%%%%%%%%%
\begin{figure}[t]
    \centering
    \includegraphics[width=0.6\linewidth]{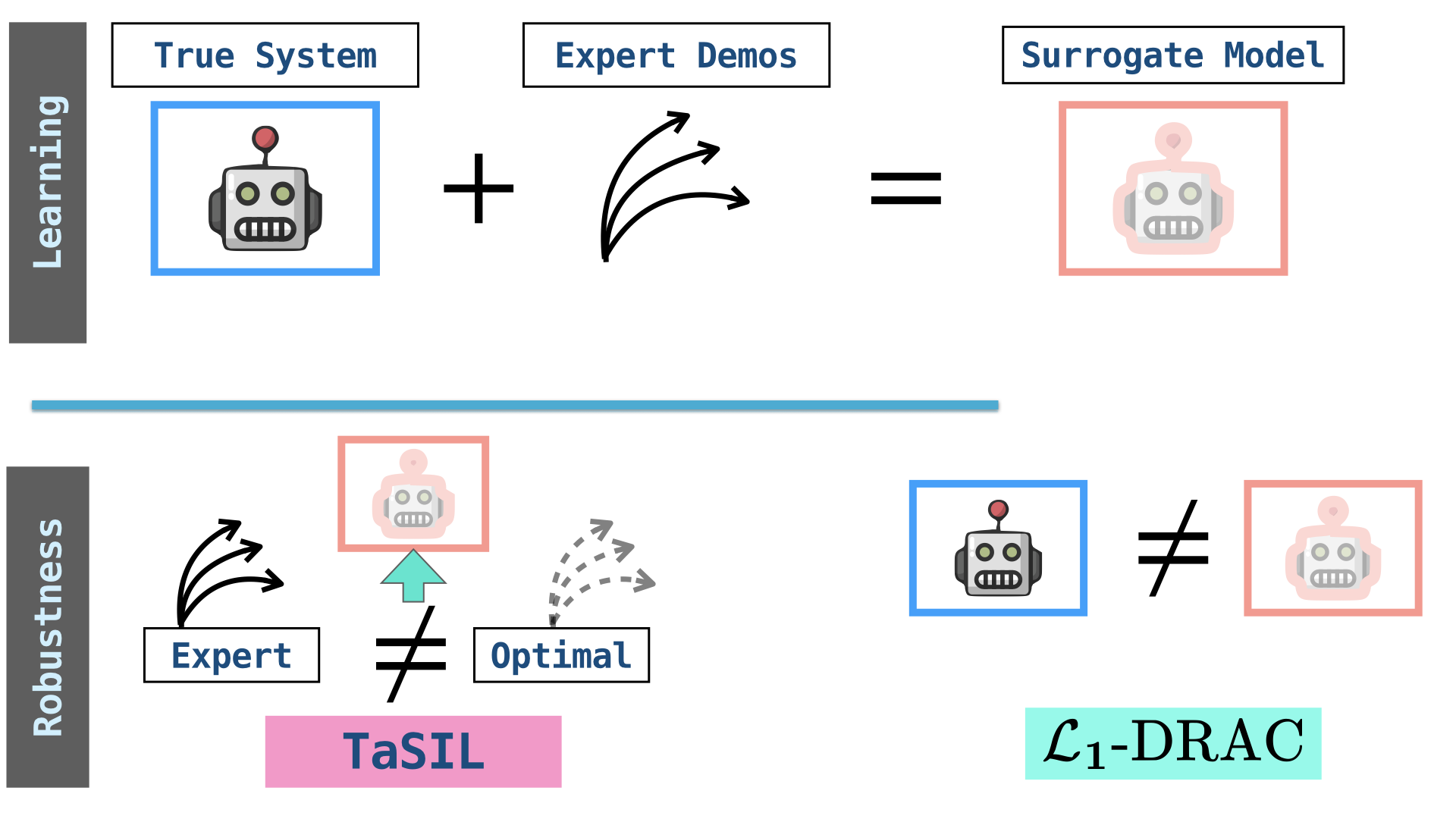}
    \caption{\emph{Top panel:} The expert trajectories available for imitation learning in our formulation are generated by the uncertain (true) system operating under some expert input process. 
    The expert trajectory data is then used to learn a nominal model whose predictive performance is only guaranteed on the expert trajectories.
    \emph{Bottom panel:} The methodology of TaSIL is designed to be robust against distribution shifts due to a difference between the expert and optimal policies. On the other hand, \elloneDRAC is designed to be robust against the effects of inaccuracies between the true and surrogate models. Together, TaSIL and \elloneDRAC can thus offer robustness guarantees against a comprehensive set of distribution shift sources.}
    \label{fig:HighLevel}
\end{figure}
%%%%%%%%%%%%%%%%%%%%%%%%%%%%%%%%%%%%%%%%%%%%%%%%%%%%%%%%%%%%
%%%%%%%%%%%%%%%%%%%%%%%%%%%%%%%
\subsection{Training Data}\label{subsec:TrainingData}
%%%%%%%%%%%%%%%%%%%%%%%%%%%%%%%

Following the standard formulation of imitation learning, we assume access to a dataset of expert demonstrations.
However, in our formulation, the expert demonstrations are generated by the uncertain (true) system operating under some unknown expert inputs. 
Thus, we assume availability of $n \in \mathbb{N}$ \emph{expert trajectories} $\cbr{ \Xt{t}\left(\xi_i,\, \Wstart{}; U^\star \right),~\xi_i \sim \mathcal{D} }_{i=1}^n$, where $U^\star$ is the \emph{unknown (expert) input process}, and $\Wstart{t}$ is a $\mathbb{P}$-Brownian motion independent of $\sigma\br{\mathcal{D}}$ the driving $\Wt{t}$ in~\eqref{eqn:L1DRACProcess}.
Using the expert trajectories, one can learn the nominal vector field $\bar{f}$ and a nominal expert policy $\pi^\star$, that define the nominal (expert) process~\eqref{eqn:ExpertProcess}.
The learning setup is illustrated in Fig.~\ref{fig:HighLevel}.

We emphasize that the \textbf{predictive performance of the learned nominal model can only be guaranteed on its training data (expert trajectories)}. 
This is a crucial aspect of our formulation as it ensures a realistic limitation that \emph{a priori} guaranteeing generalization bounds on learned models beyond their training data is infeasible without further assumptions. 
However, since we do not make any generalization assumptions, the only guarantee we can make is \emph{only} on the accuracy of the learned model on the expert trajectories, usually given by the final empirical risk. 
We assume w.l.o.g. that the empirical risk is insignificant and can thus be ignored in the subsequent analysis.
As we will see later, the bounds on the imitation gap that we provide can additively incorporate the empirical risk of the learned model and is a trivial extension. 

With the nominal model accurate on the expert trajectories, i.e., $\Xt{t}\left(\xi_i,\, \Wstart{}; U^\star \right) = \xnomt{t}{\xi_i; \pi^\star}$, $\xi_i \sim \mathcal{D}$, $\forall i \in \cbr{1,\dots,n}$, we can now define the \emph{training data}.

%%%%%%%%%%%
\begin{definition}[Training Data]\label{Def:TrainingData}  
    We define the \textbf{training data (expert trajectories) set} as 
    \begin{align}\label{eqn:partition}
        \mathcal{S}_n 
        \doteq 
        \left\{
            \xnomt{t}{\xi_i; \pi^\star}, \, t \in \partition{k}{T}, \, \xi_i \overset{i.i.d.}{\sim} \mathcal{D}, \, i \in  \cbr{1, \dots, n}, \, n \in \mathbb{N}. 
        \right\}
        ,
    \end{align}
    where we define the \emph{partition} of the interval $[0,T]$ as 
    \begin{equation}\label{eqn:HorizonPartition}
        \partition{k}{T} \doteq \cbr{0, t_1,t_2,\dots,t_k}, \quad \text{where} \quad
        0 < t_1 < t_2 < \ldots < t_k = T, \quad \text{ for some } k \in \mathbb{N}.
    \end{equation}  
    Without loss of generality, we assume that the partition $\partition{k}{T}$ is uniform, i.e., $\Delta T \doteq t_{j} - t_{j-1} = T/k$, $\forall j \in \cbr{1,\dots,k}$.

    The samples $\cbr{\xi_i}_{i=1}^n$, are i.i.d. drawn from $\mathcal{D}$, generating the \textbf{empirical} distribution $\widehat{\mathcal{D}}_n \approx \mathcal{D}$ defined as  
    \begin{align}\label{eqn:empiricaldistribution}
        \widehat{\mathcal{D}}_n\br{B} = \frac{1}{n} \sum_{i=1}^n \diracmeasure{\xi_i}{B}, \quad B \in \Borel{\mathbb{R}^n},
    \end{align}
    where, for any $a \in \mathbb{R}^n$, $\diracmeasure{a}{\cdot}: \Borel{\mathbb{R}^n} \to \cbr{0,1}$ is the \textbf{dirac measure} at $a \in \mathbb{R}^n$~\cite[Sec.~1.4]{steyer2017probability}. 
\end{definition}
%%%%%%%%%%%

% %%%%%%%%%%%%%%%%%%%%%%%%%%%%%%%
% \subsection{Assumptions}
% %%%%%%%%%%%%%%%%%%%%%%%%%%%%%%%

% {\color{blue}{Collecting the assumptions}
% \begin{enumerate}
%     \item Accuracy of the learned model (ONLY on the expert trajectories)
%     \item as in~\cite[Sec.~2]{pfrommer2022tasil}, we define $\mathcal{D}$ to be a probability measure supported on a \textbf{compact} set $\mathcal{X} \subset \mathbb{R}^n$
% \end{enumerate}}

%%%%%%%%%%%%%%%%%%%%%%%%%%%%%%%
\subsection{Problem Statement}
%%%%%%%%%%%%%%%%%%%%%%%%%%%%%%%

We wish to establish the distribution gap between the nominal (expert) system's trajectories initialized over the distribution $\mathcal{D}$ (equivalently, the uncertain (expert) system, as defined in Sec.~\ref{subsec:TrainingData}), and the uncertain (true) system's trajectories initialized over an arbitrary distribution $\bar{\mathcal{D}}$, and with the \ellone-TaSIL policy $\pi_{ad} = \fdbktasil + \fdbkellone$.
% Here, we have set $\fdbktasil = \pi^\star$.   

To this end, we define the following joint process composed of the nominal (expert) and uncertain (\ellonedrac) systems.
%%%%%%%%%%%%%%%%%%%%%%%%%%%
\begin{definition}[TaSIL - \ellonedrac Error Process]
    We say that $\Yt{t} \in \mathbb{R}^{2n}$, with $t \in [0,T]$, is the joint nominal (expert)---uncertain (\ellone-TaSIL) process---if it is the strong solution to the following:
    \begin{align}
        d\Yt{t} = \Gmu{t,\Yt{t},\Ut{t}}dt + \Gsigma{t,\Yt{t}}d\Wrt{t}, 
        \quad  
        \Yt{0} = \widetilde{\xi} \sim \widetilde{\mathcal{D}},  
    \end{align}
    on $\left(\Omega, \mathcal{F}, \delta_{0_n} \times \Wfilt{0,t}, \mathbb{P}\right)$, where $\widetilde{\xi} = \left(\xi, \bar{\xi}\right)$ and $\widetilde{\mathcal{D}}$ is an arbitrary coupling between the measures $\mathcal{D}$ and $\bar{\mathcal{D}}$ on $\Borel{\mathbb{R}^{2n}}$~\cite[Chpt.~1]{villani2009optimal}, and
    \begin{align*}
        \Yt{t} = \begin{bmatrix} \xstart{t} \\ \Xt{t} \end{bmatrix} \in \mathbb{R}^{2n},
        \quad 
        \Ut{t} =  \begin{bmatrix} \pistar{\xstart{t}} \\ \pi_{ad}\br{\Xt{t}} = \br{\fdbktasil+\fdbkellone}\br{\Xt{t}}  \end{bmatrix} \in \mathbb{R}^{2m}, 
        \quad 
        \Wrt{t} = \begin{bmatrix} 0_{d} \\ \Wt{t} \end{bmatrix} \in \mathbb{R}^{2d},
        \\ 
        \Gmu{t,\Yt{t},\Ut{t}} =
        \begin{bmatrix} 
            \fbar{t,\xstart{t}, \pistar{\xstart{t}}} 
            \\ 
            \Fmu{t,\Xt{t}, \pi_{ad}\br{\Xt{t}} }
        \end{bmatrix} \in \mathbb{R}^{2n},
        \quad  
        \Gsigma{t,\Yt{t}} =
        \begin{bmatrix} 
            0_{n,d} & 0_{n,d} 
            \\ 
            0_{n,d} & \Fsigma{t,\Xt{t}}
        \end{bmatrix} \in \mathbb{R}^{2n \times 2d},
    \end{align*}
\end{definition}

%%%%%%%%%%%%%%%%%%%%%%%%%%%
Next, we define the total imitation gap (TIG) for the imitation learning problem under consideration.
%%%%%%%%%%%%%%%%%%%%%%%%%%%

\begin{definition}\label{def:ImitationGap}
    Given the expert policy $\pi^\star$, we define the \textbf{total imitation gap (TIG)} for any feasible policy $\pi$ as  
    \begin{align}
        \Upsilon_T \left(\widetilde{\mathcal{D}}; \pi  \right)
        = \max_{t \in [0,T]}
        \mathbb{E}_{\widetilde{\xi} \sim \widetilde{\mathcal{D}}}
        \norm{\Xt{t}\left(\bar{\xi}; \pi \right) - \xnomt{t}{\xi; \pi^\star}}
    \end{align}
    where $\widetilde{\xi} = \left(\xi, \bar{\xi}\right)$, and $\widetilde{\mathcal{D}}$ is an arbitrary coupling of the initial measures $\mathcal{D}$ and $\bar{\mathcal{D}}$ on $\Borel{\mathbb{R}^{2n}}$.
\end{definition}
%%%%%%%%%%%%%%%%%%%%%%%%%%%

Note the general nature of the total imitation gap as the TIG quantifies the \textbf{out-of-distribution (OOD)} imitation gap. 
The distribution shift leading to the OOD imitation gap, and not just out-of-sample (within-distribution) imitation gap, arises from the presence of even a single source out of the three mentioned in Sec.~\ref{sec:Introduction}, namely; (i) initial distributional ambiguity ($\mathcal{D} \neq \bar{\mathcal{D}}$), (ii) epistemic uncertainty ($\Lambda_{\cbr{\mu,\sigma}}$), and (iii) aleatoric uncertainty (Brownian motion $\Wt{}$). 
Our problem formulation contains all three sources of distribution shifts, in addition to the policy-shift, and thus the TIG captures a significantly general notion of imitation gap.

\noindent \textbf{Problem Statement:} Given expert demonstrations $\mathcal{S}_n$, as in Definition~\ref{Def:TrainingData},  without any further assumptions on the availability of an expert oracle or a simulator of the true system, we wish to synthesize a feedback policy $\pi_{ad} = \fdbktasil + \fdbkellone$ that ensures the existence of an \emph{a priori} known $\rho \in \mathbb{R}_{>0}$ such that  
    \begin{align}\label{eqn:TIG}
        \Upsilon_T \left(\widetilde{\mathcal{D}}; \pi_{ad} = \fdbktasil + \fdbkellone  \right)
        \leq \rho.
    \end{align}

%%%%%%%%%%%%%%%%%%%%%%%%%%%%%%%%%%%%%%%%%%%%%%%%%%%%%%%%%%%%%%%%%%%%%%%%%%%%%%%%

%%%%%%%%%%%%%%%%%%%%%%%%%%%%%%%%%%%%%%%%%%%%%%%%%%%%%%%%%%%%%%%%%%%%%%%%%%%%%%%%
\section{Methodology: Layered Robustness via Control Composition}\label{sec:Methodology}
%%%%%%%%%%%%%%%%%%%%%%%%%%%%%%%%%%%%%%%%%%%%%%%%%%%%%%%%%%%%%%%%%%%%%%%%%%%%%%%%

%%%%%%%%%%%%%%%%%%%%%%%%%%%%%%%%%%%%%%%%%%%%%%%%%%%%%%%%%%%%%
\begin{figure}[h]
    \centering
    \begin{subfigure}[b]{0.30\textwidth}
        \centering
        \includegraphics[width=\textwidth]{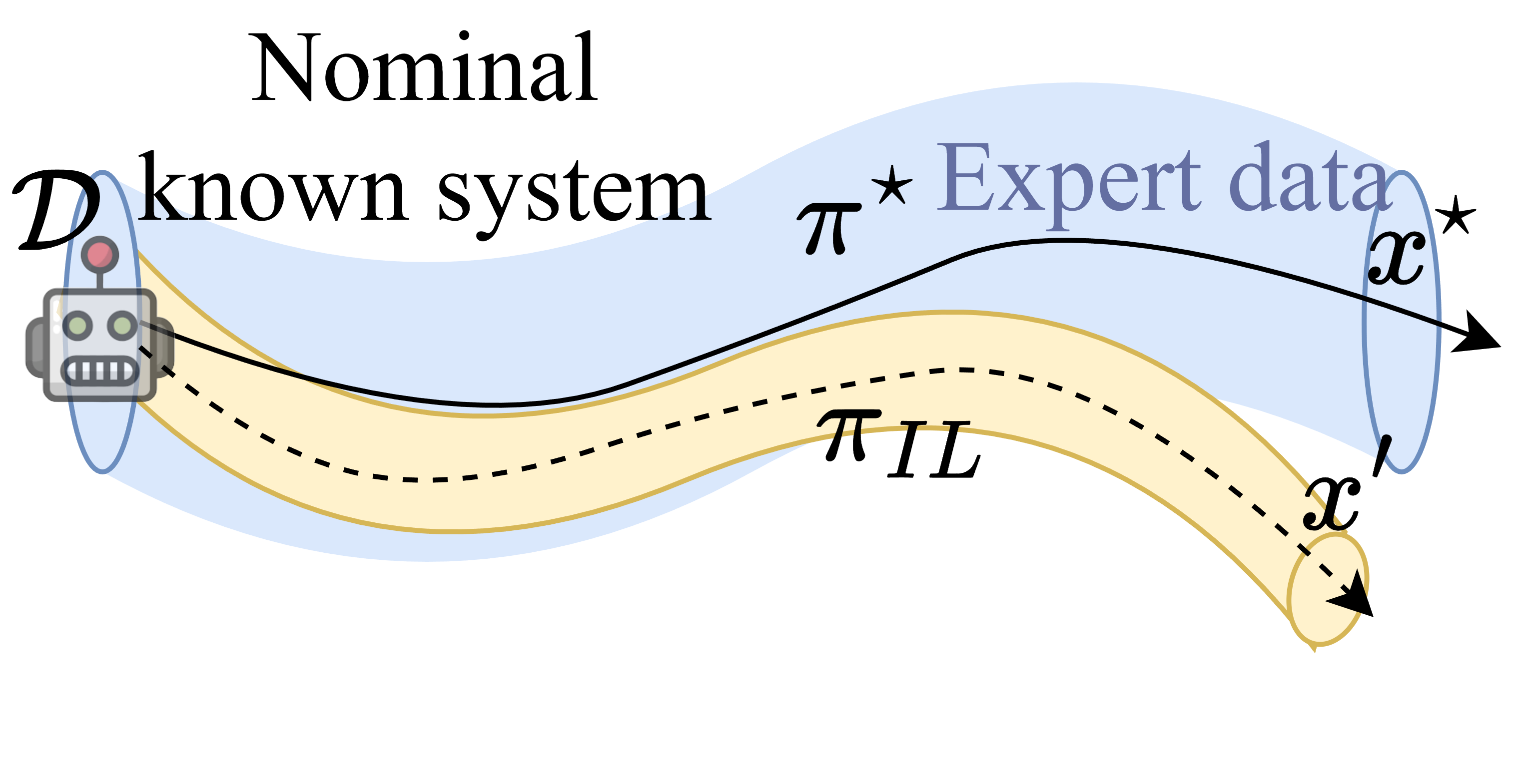}
        \caption{Policy-induced shift}
        \label{fig:dist-shift-policy}
    \end{subfigure}
    \begin{subfigure}[b]{0.30\textwidth}
        \centering
        \includegraphics[width=\textwidth]{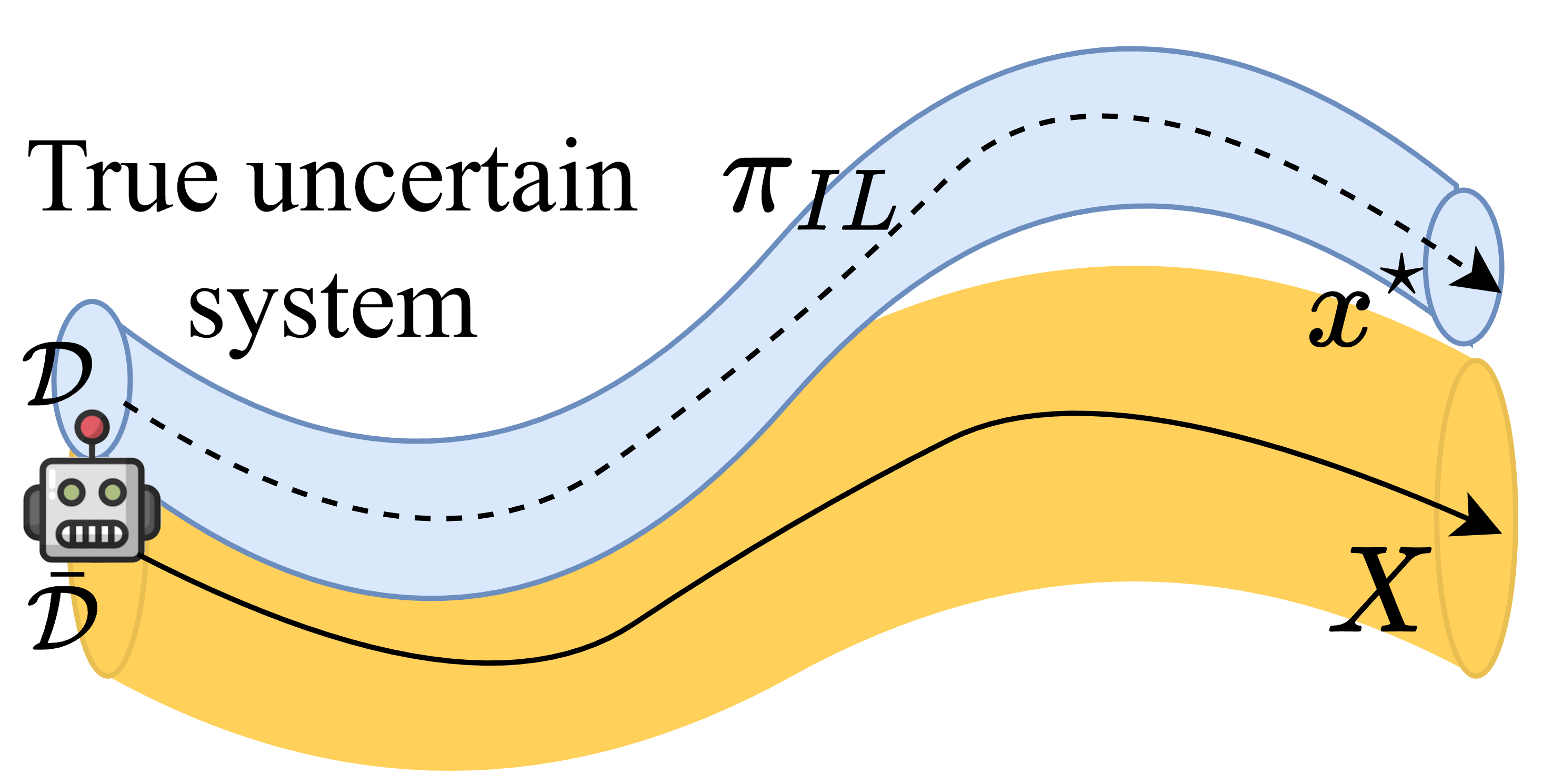}
        \caption{Uncertainty-induced shift}
        \label{fig:dist-shift-uncertainty}
    \end{subfigure}
    \begin{subfigure}[b]{0.30\textwidth}
        \centering
        \includegraphics[width=\textwidth]{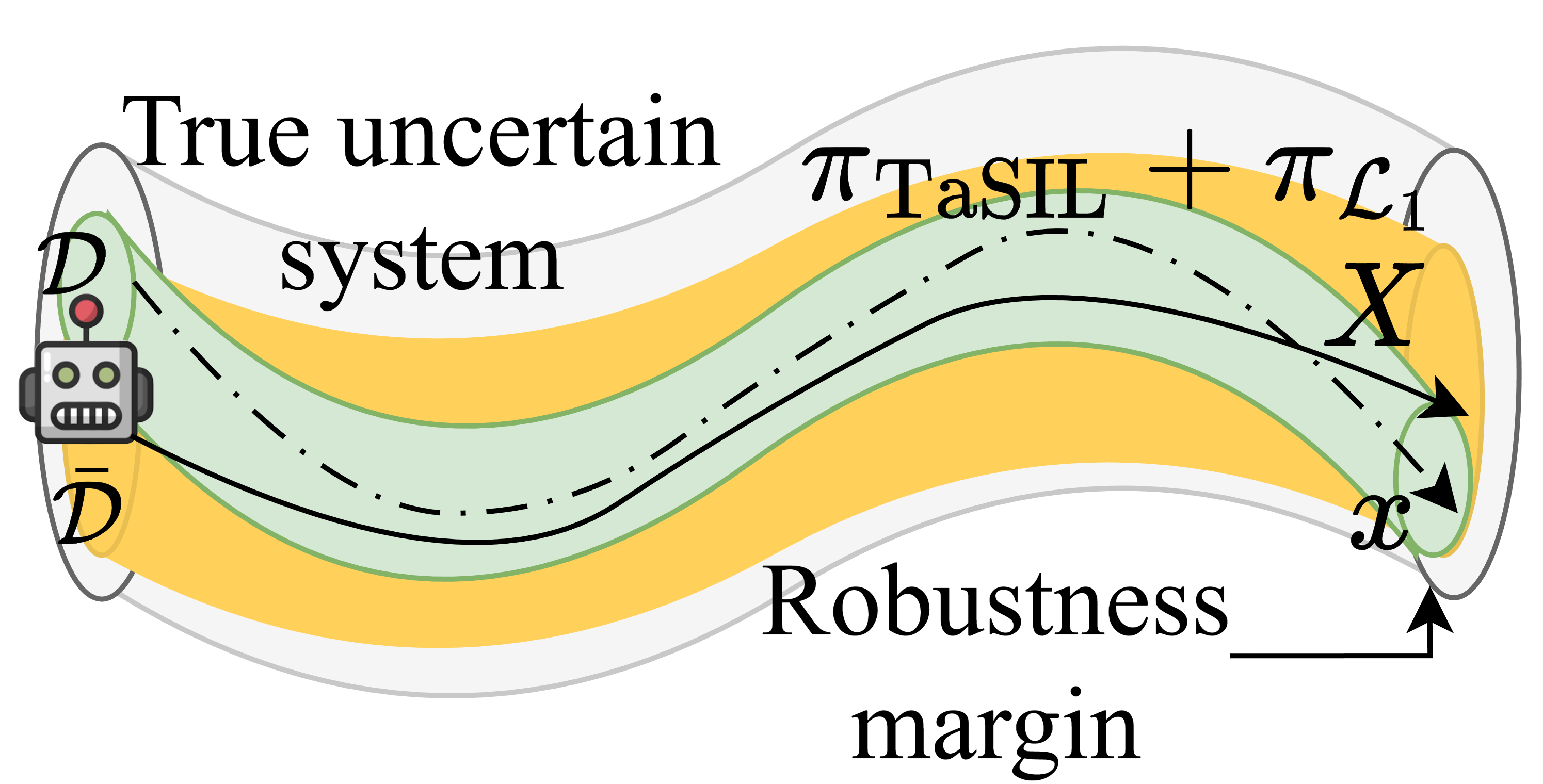}
        \caption{Decoupled approach (DRIP)}
        \label{fig:dist-shift-drip}
    \end{subfigure}
    \caption{
    Distribution shift sources and mitigation. (\subref{fig:dist-shift-policy})~Distribution shift from learned policy deviating from expert. (\subref{fig:dist-shift-uncertainty})~Distribution shift from epistemic and aleatoric uncertainties. (\subref{fig:dist-shift-drip})~Unified robustness via TaSIL + \elloneDRAC layered architecture.}
    \label{fig:decoupling}
\end{figure}
%%%%%%%%%%%%%%%%%%%%%%%%%%%%%%%%%%%%%%%%%%%%%%%%%%%%%%%%%%%%%%

We propose an LCA that integrates TaSIL and $\mathcal{L}_1$-DRAC to quantify certifiable robustness of the overall feedback operator. The proposed architecture is illustrated in Fig.~\ref{fig:LCA_TaSIL_DRAC}, which shows i) TaSIL, serving as a mid-layer planner, and ii) \elloneDRAC feedback, serving as a low-level controller. 
Although the reference commands from the mid-level TaSIL to the low-level \elloneDRAC is fairly standard, our approach also enables the converse low-level to mid-level communication in terms of  guaranteed distributional ambiguity sets. 
This bi-level communication is a key feature the enables the certifiability of the entire feedback pipeline. 

This layered, certifiably robust approach offers several benefits:
\begin{itemize}
    \item \textbf{Modularity:} Each layer can be independently designed and analyzed, promoting re-usability across tasks and platforms;
    \item \textbf{Compositional Guarantees:} Certificates from TaSIL and $\mathcal{L}_1$-DRAC can be integrated with learning-enabled components (e.g., perception modules), enabling system-wide safety and performance assurance;
    \item \textbf{Scalability and Extensibility:} The architecture supports the inclusion of high-level symbolic reasoning, perception modules, and additional learning components.
\end{itemize}

We also illustrate the high-level motivation in Fig.~\ref{fig:HighLevel}. 
While TaSIL accounts for the policy-induced distribution shifts for the learned (nominal) system, \elloneDRAC accounts for the distribution shifts due to the modeling discrepancy itself. 
Thus, by unifying the individual components we enable certifiable robustness to the wide class of sources that can lead to distribution shifts. 

%%%%%%%%%%%%%%%%%%%%%%%%%%%%%%%%%%%%%%%%%%%%%%%%%%%%%%%%%%%%
\subsection{Decoupling the Policy and Uncertainties}\label{sec:decouple_distribution_shifts}
%%%%%%%%%%%%%%%%%%%%%%%%%%%%%%%%%%%%%%%%%%%%%%%%%%%%%%%%%%%%

Our approach relies on the \textbf{decoupling} of the TIG into individual components where one ensures robustness against policy-induced distribution shifts and the other against uncertainties. 
For any admissible policy $\hat{\pi}$, we apply the Minkowski's inequality to the total imitation gap in Definition~\ref{def:ImitationGap} upon adding and subtracting $\xnomt{t}{\xi; \hat{\pi}}$:
\begin{multline}\label{eqn:Decomposition:IG:1}
    \Upsilon_T \left(\widetilde{\mathcal{D}}; \pi  \right)
        = 
        \max_{t \in [0,T]}
        \mathbb{E}_{\widetilde{\xi} \sim \widetilde{\mathcal{D}}}
        \norm{\Xt{t}\left(\bar{\xi}; \pi \right) - \xnomt{t}{\xi; \pi^\star}}
        \\
        \leq 
        \max_{t \in [0,T]}
        \mathbb{E}_{\widetilde{\xi} \sim \widetilde{\mathcal{D}}}
        \norm{\xnomt{t}{\xi; \hat{\pi}} - \xnomt{t}{\xi; \pi^\star}}
        + 
        \max_{t \in [0,T]}
        \mathbb{E}_{\widetilde{\xi} \sim \widetilde{\mathcal{D}}}
        \norm{\Xt{t}\left(\bar{\xi}; \pi \right) - \xnomt{t}{\xi; \hat{\pi}}}.
\end{multline}
Since $\widetilde{\mathcal{D}}: \Borel{\mathbb{R}^{2n}} \rightarrow [0,1]$ is a coupling of $\mathcal{D}$ and $\bar{\mathcal{D}}$, and the first expectation on the right hand side above is a function of only $\xi \sim \mathcal{D}$, we can use the marginalization property of couplings~\cite[Chpt.~1]{villani2009optimal} to conclude that  
\begin{multline*}
    \mathbb{E}_{\widetilde{\xi} \sim \widetilde{\mathcal{D}}}
        \norm{\xnomt{t}{\xi; \hat{\pi}} - \xnomt{t}{\xi; \pi^\star}}
    = 
    \int_{\mathbb{R}^{n}}
    \int_{\mathbb{R}^{n}}    
        \norm{\xnomt{t}{\xi; \hat{\pi}} - \xnomt{t}{\xi; \pi^\star}}
    \widetilde{\mathcal{D}}\br{d\xi,d\bar{\xi} \, }
    \\
    = 
    \int_{\mathbb{R}^{n}}    
        \norm{\xnomt{t}{\xi; \hat{\pi}} - \xnomt{t}{\xi; \pi^\star}}
    \mathcal{D}\br{d\xi}
    = 
    \mathbb{E}_{\xi \sim \mathcal{D}}
        \norm{\xnomt{t}{\xi; \hat{\pi}} - \xnomt{t}{\xi; \pi^\star}}.
\end{multline*}
Substituting into~\eqref{eqn:Decomposition:IG:1} leads to the following decomposition of the \textbf{total imitation gap (TIG)}:
\begin{align}\label{eqn:Decomposition:IG:Final}
    \Upsilon_T \left(\widetilde{\mathcal{D}}; \pi  \right)
    \leq 
    \widehat{\Upsilon}_T \left(\mathcal{D}; \hat{\pi}  \right)
    +
    \widetilde{\Upsilon}_T \left(\widetilde{\mathcal{D}}; \pi, \hat{\pi}  \right).
\end{align}
where the \textbf{\emph{policy shift induced imitation gap (policy-IG)} $\widehat{\Upsilon}_T \left(\mathcal{D}; \hat{\pi}  \right)$} and the \textbf{\emph{uncertainty induced imitation gap (uncertainty-IG)} $\widetilde{\Upsilon}_T \left(\widetilde{\mathcal{D}}; \pi, \hat{\pi}  \right)$} are respectively defined as:
\begin{subequations}\label{eqn:Decomposition:IG:Components}
    \begin{align}
        \widehat{\Upsilon}_T \left(\mathcal{D}; \hat{\pi}  \right)
        =
        \max_{t \in [0,T]}
        \mathbb{E}_{\xi \sim \mathcal{D}}
        \norm{\xnomt{t}{\xi; \hat{\pi}} - \xnomt{t}{\xi; \pi^\star}},&
        \quad \text{(policy-IG)}
        \label{eqn:Decomposition:IG:Policy}
        \\
        \widetilde{\Upsilon}_T \left(\widetilde{\mathcal{D}}; \pi, \hat{\pi}  \right)
        = 
        \max_{t \in [0,T]}
        \mathbb{E}_{\widetilde{\xi} \sim \widetilde{\mathcal{D}}}
        \norm{\Xt{t}\left(\bar{\xi}; \pi \right) - \xnomt{t}{\xi; \hat{\pi}}}.&
        \quad \text{(uncertainty-IG)}
        \label{eqn:Decomposition:IG:Uncertainty}
    \end{align}
\end{subequations}
Enabled by the decoupled formulation, we bound the individual components in~\eqref{eqn:Decomposition:IG:Components} using TaSIL and \elloneDRAC methodologies, and in turn establish the robustness guarantees for the overall \ellone-TaSIL feedback policy $\pi_{ad} = \fdbktasil + \fdbkellone$ given by:
\begin{align*}
    \Upsilon_T \left(\widetilde{\mathcal{D}}; \pi_{ad}  \right)
    \leq 
    \widehat{\Upsilon}_T \left(\mathcal{D}; \fdbktasil  \right)
    +
    \widetilde{\Upsilon}_T \left(\widetilde{\mathcal{D}}; \pi_{ad}, \fdbktasil  \right), 
    \quad 
    \pi_{ad} = \fdbktasil + \fdbkellone.
\end{align*}
For our analysis, we assume that the nominal system is contracting under the expert policy. 
To do so, we need the following definition of the the \emph{induced logarithmic norm} of $A \in \mathbb{R}^{n \times n}$:
    \begin{align*}
        \mu\left(A\right) = \lim_{h \rightarrow 0^+} \frac{\norm{\mathbb{I}_n + hA} - 1}{h}.        
    \end{align*}
The induced logarithmic norm (\emph{log norm} for brevity) is guaranteed to exist and can be interpreted as the derivative of $e^{At}$ in the direction of $A$ and evaluated at $\mathbb{I}_n$.
Refer to~\cite[Sec.~2.3]{bulloContractionBook} for further details.  
%%%%%%%%%%%%%%%
\begin{assumption}[Nominal (Expert) System Stability]\label{assmp:ILF}
    There exists a known $\lambda \in \mathbb{R}_{>0}$ such that 
    \begin{align}\label{eqn:ILFConditions} 
        \mu\left(
           \nabla_\zeta \fbar{t, \zeta, \pistar{\zeta}}  
        \right)
        \leq -\lambda, \quad \forall (t,\zeta) \in \mathbb{R}_{\geq 0} \times \mathbb{R}^n,
    \end{align}
    where $\bar{f}$ and $\pi^\star$ are introduced in Definitions~\ref{def:VectorFields} and~\ref{def:MainSystems}, respectively.
\end{assumption}
%%%%%%%%%%%%%%%
The condition in~\eqref{eqn:ILFConditions} allows one to use contraction theory to conclude the \emph{incremental exponential stability (IES)} of the nominal (expert) trajectories $\xnomt{t}{\xi; \pi^\star} $~\cite[Thm.~3.9]{bulloContractionBook}.

%%%%%%%%%%%%%%%%%%%%%%%%%%%%%%%%%%%%%%%%%%%%%%%%%%%%%%%%%%%%
\subsection{Bounding the Policy Shift Induced Imitation Gap}\label{sec:PolicyIG}
%%%%%%%%%%%%%%%%%%%%%%%%%%%%%%%%%%%%%%%%%%%%%%%%%%%%%%%%%%%%

We establish the bounds on the policy shift induced imitation gap (policy-IG), defined in~\eqref{eqn:Decomposition:IG:Policy}, using our previously developed approach in~\cite{pfrommer2022tasil}. 
We begin by considering a single expert trajectory, initialized at some \textbf{fixed} $\xi_0 \in \widehat{\mathcal{D}}_n$, over which the policy-IG for any admissible $\hat{\pi}$ can be written as   
\begin{align}\label{eqn:policy-IG:trajectory}
    \widehat{\Upsilon}_T \left(\xi_0; \hat{\pi}  \right)
    =
    \max_{t \in [0,T]}
    \norm{\xnomt{t}{\xi_0; \hat{\pi}} - \xnomt{t}{\xi_0; \pi^\star}}.
\end{align}
To analyze the trajectory-level policy-IG above, we make use of the contraction property of the nominal (expert) system in Assumption~\ref{assmp:ILF}.
We have the following straightforward result.
%%%%%%%%%%%%%%%
\begin{proposition}\label{prop:deltaISS}
    Under Assumption~\ref{assmp:ILF}, the nominal (expert) system~\eqref{eqn:ExpertProcess} is incrementally input-to-state stable ($\delta$-ISS)~\cite{angeli2002lyapunov}. 
    In particular, for any $\xi_1, \xi_2 \in \mathbb{R}^n$ and $\varsigma \in \mathcal{C}\left([0,T];\mathbb{R}^m\right)$ 
    \begin{align*}
        \norm{
            x_t^{\pi^\star}\left(\xi_1;\varsigma_t\right)
            -
            x_t^{\pi^\star}\left(\xi_2;0_m\right)
        }
        \leq
        e^{-\lambda_\theta t/2 } \norm{\xi_1 - \xi_2}
        + 
        \frac{1}{\sqrt{\theta \lambda_\theta}} 
        \left( 1 - e^{-\lambda_\theta t} \right)^\frac{1}{2}
        \sup_{s \in [0,T]} \norm{\varsigma_s},
        \quad \forall t \in [0,T],
    \end{align*}
     where $\lambda_\theta = 2 \lambda - \theta \Delta_g^2$, $\theta \in \left(0, 2\lambda / \Delta_g^2 \right)$, and $y_t^\star = x_t^{\pi^\star}\left(\xi_1;\varsigma_t\right)$ is the solution to the the input-perturbed nominal (expert)  
    \begin{align*}
        y_t^\star 
        = 
        \fbar{t,y_t^\star, \pistar{y_t^\star} + \varsigma_t},
        \quad y_0^\star = \xi_1. 
    \end{align*}
\end{proposition}
%%%%%%%%%%%%%%%%
%%%%%%%%%%%%%%%%
\begin{proof}
    Consider the Lyapunov functional $V(t)=e(t)^\top e(t)$, where 
    \begin{align*}
        e_t = y_t^\star - x_t^\star = x_t^{\pi^\star}\left(\xi_1;\varsigma_t\right)-x_t^{\pi^\star}\left(\xi_2;0_m\right), 
        \quad 
        \xstart{t} =x_t^{\pi^\star}\left(\xi_2;0_m\right) = x_t\left(\xi_2;\pi^\star\right).
    \end{align*} 
    It then follows that  
    \begin{multline}\label{eqn:Prop1:1}
        \dot{V}(t)
        = 
        2 e(t)^\top 
        \left[ 
            \fbar{t,y_t^\star, \pistar{y_t^\star} + \varsigma_t} - \fbar{t, \xstart{t}, \pistar{\xstart{t}}} 
        \right]
        \\
        =
        2 e(t)^\top 
        \left[ 
            \fbar{t,y_t^\star, \pistar{y_t^\star}} - \fbar{t, \xstart{t}, \pistar{\xstart{t}}} 
        \right]
        +
        2 e(t)^\top 
            g(t) \varsigma_t. 
    \end{multline}
    Due to the contraction property in Assumption~\ref{assmp:ILF}, we invoke~\cite[Thm.~29]{DavydovJafarpourBullo2022} to conclude that  
    \begin{multline*}
        e(t)^\top 
        \left[ 
            \fbar{t,y_t^\star, \pistar{y_t^\star}} - \fbar{t, \xstart{t}, \pistar{\xstart{t}}} 
        \right]
        \\
        =
        \left(y_t^\star - x_t^\star\right)^\top 
        \left[ 
            \fbar{t,y_t^\star, \pistar{y_t^\star}} - \fbar{t, \xstart{t}, \pistar{\xstart{t}}} 
        \right]
        \leq 
        -\lambda \norm{y_t^\star - x_t}^2 = -\lambda V(t).
    \end{multline*}
    Substituting the above into~\eqref{eqn:Prop1:1} yields
    \begin{align*}
        \dot{V}(t) 
        \leq
        - 2 \lambda V(t) 
        +
        2 e(t)^\top 
            g(t) \varsigma_t
        \overset{(\star)}{\leq}
        - \lambda_\theta  V(t) + \frac{1}{\theta}  \norm{\varsigma_t}^2 
        % \\
        \leq
        -\lambda_\theta V(t) + \frac{1}{\theta}  \sup_{s \in [0,T]}\norm{\varsigma_s}^2, \quad \forall t \in [0,T],
    \end{align*}
    where $(\star)$ is due to following bound which results from Assumption~\ref{assmp:VectorFields} and the Young's inequality for $\theta \in \mathbb{R}_{>0}$:
    \begin{align*}
        \norm{e(t)^\top g(t) \varsigma_t} 
        \leq 
        \norm{e(t)} \Frobenius{g(t)} \norm{\varsigma_t}
        \leq 
        \Delta_g \norm{e(t)}  \norm{\varsigma_t}
        \leq 
        \frac{\theta}{2} \Delta_g^2 \norm{e(t)}^2 + \frac{1}{2\theta}  \norm{\varsigma_t}^2.
    \end{align*} 
    Since $\lambda_\theta = 2 \lambda - \theta \Delta_g^2 > 0$, we apply the special case of the comparison lemma~\cite{Khalil2002NonlinearSystems} in~\cite[Lem.~2.1]{TsukamotoChungSlotine2021} to conclude that  
    \begin{align*}
        V(t) \leq 
        e^{-\lambda_\theta t} V(0)
        + 
        \frac{1}{\theta \lambda_\theta} 
        \left( 1 - e^{-\lambda_\theta t} \right)
        \sup_{s \in [0,T]}\norm{\varsigma_s}^2, \quad \forall t \in [0,T].
    \end{align*}
    The proof is then concluded by using the definition $V(t) = \norm{e(t)}^2$ and the sub-additivity of the square root operator. 
\end{proof}
%%%%%%%%%%%%%%%%
The result in Proposition~\ref{prop:deltaISS} further translates to the discrete-time setting, for which, we note that we can write
\begin{align}\label{eqn:DiscreteNominal}
    x_{t+\Delta T} 
    = 
    \Phi_{\Delta T}\br{t, x_t, \pi^\star(x_t)}
    = 
    x_t + \int_{t}^{t+\Delta T} \fbar{s, x_s, \pi^\star(x_s)} ds
    , \quad x_0 = \xi_0, \quad t \in \partition{k}{T}/\{T\}.
\end{align}
where $\Phi_{\Delta T}:\mathbb{R}_{>0} \times \mathbb{R}^n \times \mathbb{R}^m \rightarrow \mathbb{R}^n$ is the flow map of the nominal system, and the partition $\partition{k}{T}$ of the interval $[0,T]$ is defined in~\eqref{eqn:HorizonPartition}.
As is evident , the discrete-time nominal system in~\eqref{eqn:DiscreteNominal} is simply the samples of the continuous-time nominal trajectory at the elements of the partition $\partition{k}{T}/\{T\}$.
Then, if we assume that $\varsigma_t$ is a piecewise constant signal over the partition $\partition{k}{T}$, i.e., $\varsigma_t = \varsigma_{t_i}$, for all $t \in [t_i, t_{i+1})$, $i \in \{0, \ldots, k-1\}$, we have the following corollary to Proposition~\ref{prop:deltaISS}, which we state without proof. 
%%%%%%%%%%%%%%%%
\begin{corollary}\label{corol:DiscreteDeltaISS}
    Under Assumption~\ref{assmp:ILF}, the discrete-time nominal (expert) system~\eqref{eqn:ExpertProcess} satisfies
    \begin{align*}
        \norm{
            x_t^{\pi^\star}\left(\xi_1;\varsigma_t\right)
            -
            x_t^{\pi^\star}\left(\xi_2;0_m\right)
        }
        \leq
        e^{-\lambda_\theta t/2 } \norm{\xi_1 - \xi_2}
        + 
        \frac{1}{\sqrt{\theta \lambda_\theta}} 
        \left( 1 - e^{-\lambda_\theta t} \right)^\frac{1}{2}
        \max_{s \in \partition{k}{T}/\{T\}} \norm{\varsigma_s},
    \end{align*}
    for all $t \in \partition{k}{T}/\{T\}$, and for any $\xi_1, \xi_2 \in \mathbb{R}^n$ and $\cbr{\varsigma_{t_i} \in \mathbb{R}^m}_{i=0}^{k-1}$, where $\lambda_\theta, \theta \in \mathbb{R}_{>0}$ are defined in Proposition~\ref{prop:deltaISS}.    
\end{corollary}    
%%%%%%%%%%%%%%%% 
Now, under the test policy $\hat{\pi}$, the trajectory $x^{\hat{\pi}}_t(\xi) = \xnomt{t}{\xi; \hat{\pi}}$ satisfies 
\begin{align*} 
    x^{\hat{\pi}}_{t+\Delta T}(\xi)
    = 
    \Phi_{\Delta T}\br{t, x^{\hat{\pi}}_t(\xi) , \hat{\pi}\left(x^{\hat{\pi}}_t(\xi)\right)}
    = 
    \Phi_{\Delta T}\br{t, x^{\hat{\pi}}_t(\xi) , \pi^\star \left(x^{\hat{\pi}}_t(\xi)  \right) + \Theta^{\hat{\pi}}_t(\xi,\hat{\pi})}, \quad t \in \partition{k}{T}/\{T\},
\end{align*}
where 
\begin{align}\label{eqn:InputPerturbation}
    \Theta^{\hat{\pi}}_t(\xi,\hat{\pi}) = \hat{\pi}\left(x^{\hat{\pi}}_t(\xi)\right) - \pi^\star \left(x^{\hat{\pi}}_t(\xi)  \right) \in \mathbb{R}^m, 
    \quad \quad t \in \partition{k}{T}/\{T\}.
\end{align}
Using the notation introduced in Proposition~\ref{prop:deltaISS}, the previous expression implies that
\begin{align*}
    \xnomt{t}{\xi; \hat{\pi}} = x_t^{\pi^\star}\left(\xi; \Theta^{\hat{\pi}}_t(\xi,\hat{\pi})\right),
\end{align*} 
which in turn allows us to write the policy-IG in~\eqref{eqn:policy-IG:trajectory} as
\begin{align}\label{eqn:policy-IG:trajectory:Perturbed}
    \widehat{\Upsilon}_T \left(\xi; \hat{\pi}  \right)
    =
    \max_{t \in \partition{k}{T}/\{T\}}
    \norm{x_t^{\pi^\star}\left(\xi; \Theta^{\hat{\pi}}_t(\xi,\hat{\pi})\right) - x_t^{\pi^\star}\left(\xi; 0_m \right)  },
\end{align}
where we have used the definition $\xnomt{t}{\xi; \pi^\star} = x_t^{\pi^\star}\left(\xi; 0_m \right)$.
%%%%%%%%%%%%%%%%
\begin{remark}
    Compared to the definition in~\eqref{eqn:policy-IG:trajectory}, in the re-written expression for the policy-IG above, the maximum is taken over the discrete-time partition $\partition{k}{T}/\{T\} = \cbr{0,t_1, \ldots, t_{k-1}}$. 
    The discrete-time version can be developed into the continuous-time bound via Lipschitz continuity arguments of the feedback policies. 
    Additionally, the subsequent analysis holds true for the continuous-time setting \emph{mutatis mutandis}, given that the $\delta$-ISS property in Proposition~\ref{prop:deltaISS} is originally developed for the continuous-time setting.
    Since this is a technical detail, we omit it w.l.o.g for brevity. 
\end{remark}
%%%%%%%%%%%%%%%%
Now, it follows from Corollary~\ref{corol:DiscreteDeltaISS} that 
\begin{align}
    \norm{x_t^{\pi^\star}\left(\xi; \Theta^{\hat{\pi}}_t(\xi,\hat{\pi})\right) - x_t^{\pi^\star}\left(\xi; 0_m \right)  }
    \leq
    \frac{1}{\sqrt{\kappa \lambda_\kappa}} 
    \left( 1 - e^{-\lambda_\kappa t} \right)^\frac{1}{2}
    \max_{s \in \partition{k}{T}/\{T\}} \norm{\Theta^{\hat{\pi}}_s(\xi,\hat{\pi})},
\end{align}
for all $t \in \partition{k}{T}/\{T\}$, which leads us to the conclusion in our previous work~\cite{pfrommer2022tasil} that the trajectory-level policy-IG is controlled by the maximum deviation between the test policy $\hat{\pi}$ and the expert policy $\pi^\star$ along the trajectory generated by $\hat{\pi}$. 

We place the following regularity assumption on the policies up to their first-order derivatives.  
%%%%%%%%%%%%%%%%%%%%%%%%%%%%%%%%%
\begin{assumption}\label{assmp:PolicyRegularity}
    There exist known $L_\pi, L_{\partial \pi} \in \mathbb{R}_{>0}$ such that for $\bar{\pi} \in \cbr{\hat{\pi},\pi^\star}$ 
    \begin{subequations}
        \begin{align}
           \norm{\bar{\pi}(\zeta_1) - \bar{\pi}(\zeta_2)} \leq L_\pi \norm{\zeta_1 - \zeta_2},  \quad \forall \left(\zeta_1, \zeta_2\right) \in \mathbb{R}^n \times \mathbb{R}^n,
           \\
           \norm{\bar{\pi}(\zeta) - \left(\bar{\pi}(\zeta_0)^\top + \nabla \bar{\pi}(\zeta_0)^\top (\zeta - \zeta_0)\right)} 
           \leq   
           \frac{L_{\partial \pi}}{2} \norm{\zeta - \zeta_0}^2, \quad \forall \left(\zeta, \zeta_0\right) \in \mathbb{R}^n \times \mathbb{R}^n.
        \end{align}
    \end{subequations}
\end{assumption}
%%%%%%%%%%%%%%%%%%%%%%%%%%%%%%%%%

With the setup above, we can now directly appeal to our previous work~\cite{pfrommer2022tasil} under a few observations.  
The $\delta$-ISS property in Corollary~\ref{corol:DiscreteDeltaISS} implies the robustness of the nominal system governed by the class-$\mathcal{K}$ function $Cx$, $C = \frac{1}{\sqrt{\theta \lambda_\theta}} \left( 1 - e^{-\lambda_\theta t} \right)^\frac{1}{2}$.  
For such nominal systems, we can set $p=1$ and $r=1$ in~\cite{pfrommer2022tasil}, thus we place the assumption on the first order ($p$=1) Taylor polynomial in Assumption~\ref{assmp:PolicyRegularity}. 
That is, we match the derivatives of the test policy $\hat{\pi}$ to the expert policy $\pi^\star$ up to the first order.
While we can choose $p>1$, the choice offers a tradeoff that is unnecessary for our approach.
Please see the discussion post Thm.~$3.2$ in~\cite{pfrommer2022tasil} for more details on the choice of $p$ and its relation to the growth rate of the class-$\mathcal{K}$ governing the $\delta$-ISS property.
In order to obtain the TaSIL policy $\fdbktasil$, one solves the following: 
\begin{align}\label{eqn:LossTaSIL}
    \fdbktasil = \argmin_{\hat{\pi} \in \Pi} \mathbb{E}_{\xi \sim \widehat{\mathcal{D}}_n}
        l^{\pi^\star}\br{\xi,\hat{\pi}}, 
    \quad  
    l^{\pi^\star}\br{\xi,\hat{\pi}} 
    = \frac{1}{2} 
    \left(  
        \max_{s \in \partition{k}{T}/\{T\}} \norm{ \Psi^{\pi^\star}_t(\xi,\hat{\pi})}
        + 
        \max_{s \in \partition{k}{T}/\{T\}} \norm{ \nabla_x \Psi^{\pi^\star}_t(\xi,\hat{\pi})}     
    \right),
\end{align} 
where, $\Psi^{\pi^\star}_t(\xi,\hat{\pi})$ is the input perturbation defined in~\eqref{eqn:InputPerturbation} but with the rollout policy $\pi^\star$
\begin{align*}
    \Psi^{\pi^\star}_t(\xi,\hat{\pi}) = \hat{\pi}\left(x^{\pi^\star}_t(\xi)\right) - \pi^\star \left(x^{\pi^\star}_t(\xi)  \right) \in \mathbb{R}^m, 
    \quad \quad t \in \partition{k}{T}/\{T\}.
\end{align*}
It is therefore evident that the loss function in~\eqref{eqn:LossTaSIL} is evaluated only on the training data without any further requirements.  
We now state the abbreviated result from~\cite{pfrommer2022tasil} for the systems considered in the manuscript.
Please see~\cite[Thm.~4.3]{pfrommer2022tasil} for the details. 
%%%%%%%%%%%%%%%%%
\begin{theorem}\label{thm:TaSIL}
    Fix a failure probability $\delta \in (0,1)$.
    The policy-IG $\widehat{\Upsilon}_T \left(\mathcal{D}; \hat{\pi}  \right)$, defined in~\eqref{eqn:Decomposition:IG:Policy}, evaluated at $\fdbktasil$ that is obtained by solving the optimization problem in~\eqref{eqn:LossTaSIL}, satisfies
    \begin{align}
        \widehat{\Upsilon}_T \left(\mathcal{D}; \fdbktasil  \right)
        \leq 
        \rho_{\text{\tiny TaSIL}}(\delta,n),
    \end{align}
    where the constant $\rho_{\text{\tiny TaSIL}}(\delta,n)$ is a computable constant that satisfies 
    \begin{align}
         \rho_{\text{\tiny TaSIL}}(\delta,n) \in \mathcal{O}\left(\frac{\log n}{n}\right), \text{ as } n \uparrow \infty, \quad \text{and} \quad
        \rho_{\text{\tiny TaSIL}}(\delta,n) \in \mathcal{O}\left(\frac{1}{\delta}\right), \text{ as } \delta \downarrow 0.    
    \end{align}  
\end{theorem}
%%%
\begin{proof}
    The results follow from~\cite[Thm.~4.3]{pfrommer2022tasil} with $p=1$ and $r=1$.
    In particular, for $p=1$,~\cite[Thm.~4.3]{pfrommer2022tasil} proves that   
    \begin{align*}
        \Gamma_T \left(\xi; \pi_{\text{TaSIL},1}  \right) \leq \rho_{\text{\tiny TaSIL}}(\delta,n),
    \end{align*}
    where $\Gamma_T$ denotes the imitation gap.
    As evident from~\eqref{eqn:LossTaSIL}, the imitation gap above is for the $\fdbktasil$ trained to minimize $\mathbb{E}_{\xi \sim \widehat{\mathcal{D}}_n}\sbr{\max \br{\text{trajectory deviation}}}$, which corresponds to minimization of the pathwise/trajectory-level law (distribution) of  policy-IG.
    However, the policy-IG in \eqref{eqn:Decomposition:IG:Policy} that we consider in the manuscript is the maximum over the instantaneous (pointwise in time) distributions induced by the trajectories under the policy-shift.  
    The former implies the latter, and thus the result follows.
    Note that this equivalence does not hold in the other direction.
\end{proof}
%%%%%%%%%%%%%%%%%

%%%%%%%%%%%%%%%%%%%%%%%%%%%%%%%%%%%%%%%%%%%%%%%%%%%%%%%%%%%%
\subsection{Bounding the Uncertainty Induced Imitation Gap}\label{sec:UncertaintyIG}
%%%%%%%%%%%%%%%%%%%%%%%%%%%%%%%%%%%%%%%%%%%%%%%%%%%%%%%%%%%%}

We now establish the bounds on the uncertainty induced imitation gap (uncertainty-IG), defined in~\eqref{eqn:Decomposition:IG:Policy} as 
\begin{align}
    \widetilde{\Upsilon}_T \left(\widetilde{\mathcal{D}}; \pi_{ad}, \fdbktasil  \right)
    = 
    \max_{t \in [0,T]}
    \mathbb{E}_{\widetilde{\xi} \sim \widetilde{\mathcal{D}}}
    \norm{\Xt{t}\left(\bar{\xi}; \pi_{ad} \right) - \xnomt{t}{\xi; \fdbktasil}}, \quad 
    \pi_{ad} = \fdbktasil + \fdbkellone,
\end{align}
where $\fdbktasil$ is the imitation feedback policy obtained in Sec.~\ref{sec:PolicyIG}, and $\widetilde{\xi} = \left(\xi, \bar{\xi}\right)$, and $\widetilde{\mathcal{D}}$, arbitrary coupling of the initial measures $\mathcal{D}$ and $\bar{\mathcal{D}}$ on $\Borel{\mathbb{R}^{2n}}$, are provided in Definition~\ref{def:ImitationGap}.
We choose the policy $\fdbkellone$ as our recently developed \elloneDRAC control laws~\cite{L1DRAC} which ensures certifiable robustness against the distribution shifts due to the aleatoric and epistemic uncertainties in the system dynamics. 
the \ellonedrac~feedback operator $\FL$ consists of a process predictor, an adaptation law, and a low-pass filter as illustrated in Fig.~\ref{fig:L1DRAC}. 
%%%%%%%%%%%%%%%%%%%%%%%%%%%%%
\begin{figure}[h]
    \centering
    \includegraphics[width=0.50\textwidth]{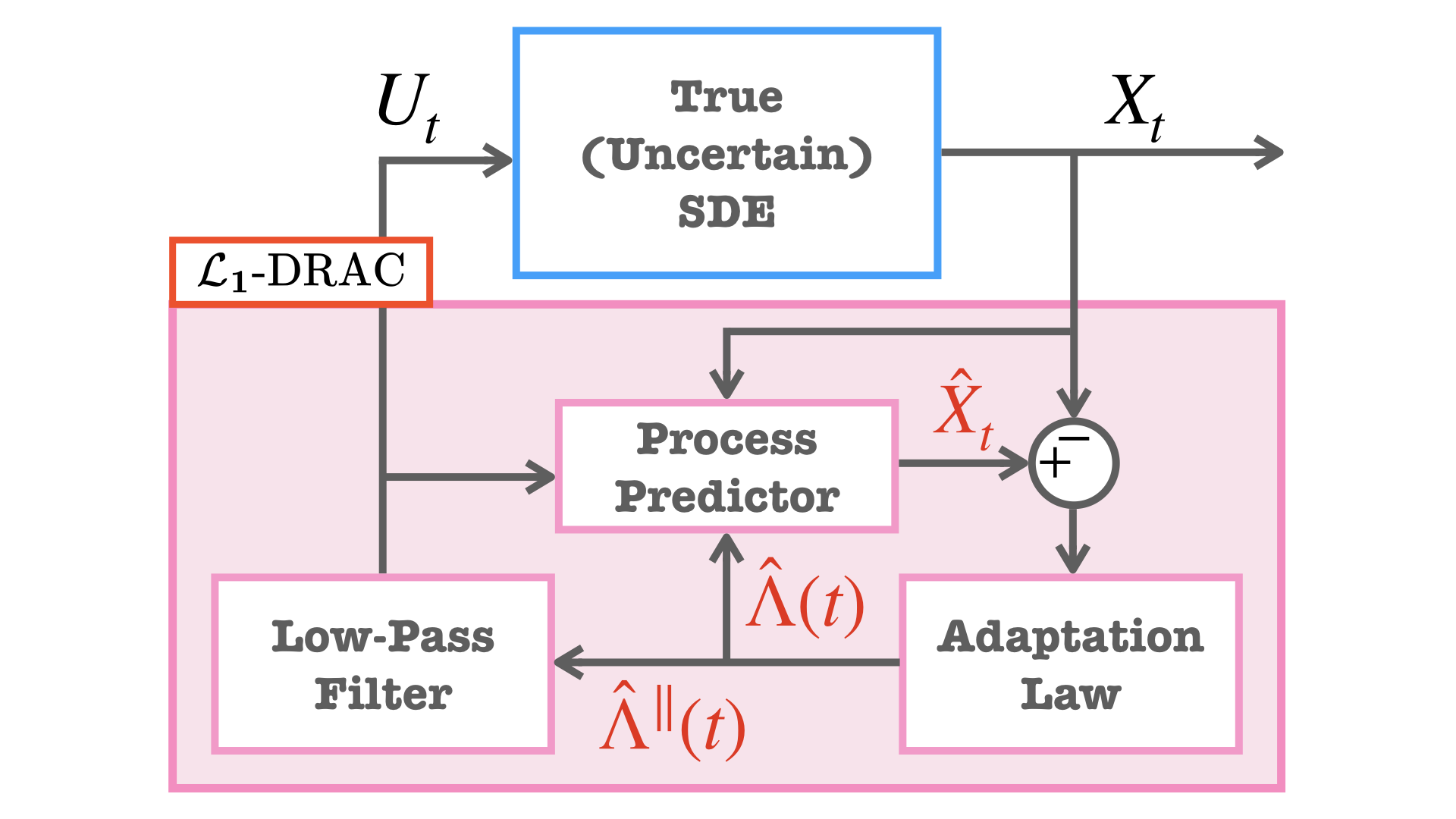}
    \caption{The architecture of the \ellonedrac~controller. The controller has three components: a \emph{process predictor} with output $\hat{X}_t$, an \emph{adaptation law}, and a \emph{low pass filter}.}
    \label{fig:L1DRAC}
\end{figure}
%%%%%%%%%%%%%%%%%%%%%%%%%%%%%

We define the \ellonedrac~feedback operator $\FL:\mathcal{C}([0,T]:\mathbb{R}^n) \rightarrow \mathcal{C}([0,T]:\mathbb{R}^m)$, as follows: 
\begin{equation}\label{eqn:L1DRAC:Definition:FeedbackOperator}
    \FL[y] \doteq \Filter \circ \AdaptationLaw \circ \Predictor[y], \quad y \in \mathcal{C}([0,T]:\mathbb{R}^n),  
\end{equation}
which we can alternatively represent as
\begin{align}\label{eqn:L1DRAC:Definition:FeedbackOperator:Alt}
    \FL[y] 
    = 
    \Filter[\hat{\Lambda}^{\paral}],
    \quad
    \hat{\Lambda}^{\paral} = \AdaptationLawParal[\hat{\Lambda}],
    \quad 
    \hat{\Lambda} = \AdaptationLaw[\hat{y}][y],
    \quad 
    \hat{y} = \Predictor[y]
    ,
\end{align} 
where
\begin{subequations}\label{eqn:L1DRAC:Definition:FeedbackOperator:Components}
    \begin{align}
        \Filter[\hat{\Lambda}^{\paral}][t] 
        \doteq  
        - \Boldomega \int_0^t \expo{-\Boldomega(t-\nu)}\Lparahat{\nu}d\nu,  \quad \text{\emph{(Low-pass filter)}}
        \label{eqn:L1DRAC:Definition:FeedbackOperator:Filter}
        \\
        \begin{aligned}
            \Lparahat{t} 
            =&
            \AdaptationLawParal[\hat{\Lambda}][t]
            = 
            \sum_{i=0}^{\lfloor \frac{t}{\BoldTs} \rfloor}  
            \Theta_{ad}(i\BoldTs) \Lhat{t}
            \indicator{[i\BoldTs,(i+1)\BoldTs)}{t}
            , 
            \\
            \Lhat{t} 
            =&
            \AdaptationLaw[\hat{y}][y][t]
            =  
            0_n \indicator{[0,\BoldTs)}{t} 
            \\
            &+
            \lambda_s \br{1 - e^{\lambda_s \BoldTs}}^{-1} 
            \sum_{i=1}^{\lfloor \frac{t}{\BoldTs} \rfloor}    
            \br{\hat{y}\br{i\BoldTs} - y\br{i\BoldTs}}
            \indicator{[i\BoldTs,(i+1)\BoldTs)}{t}, 
        \end{aligned} 
        \quad \text{\emph{(Adaptation Law)}}
        \label{eqn:L1DRAC:Definition:FeedbackOperator:AdaptationLaw}
        \\
        % d\Xhatt{t} = \br{-\lambda_s \mathbb{I}_n \left(\Xhatt{t} - \Xt{t} \right)+ f(t,\Xt{t}) +  g(t)\ULt{t} + \Lhat{t}} dt, \quad \Xhatt{0} = 0_n,
        \begin{aligned}
            \hat{y}(t) 
            =&
            \Predictor[y][t] \Rightarrow \text{ solution to the integral equation: } 
            \\
            \hat{y}(t)=& x_0 +  
            \int_0^t \br{-\lambda_s \mathbb{I}_n \left(\hat{y}(\nu) - y(\nu) \right)+ f(\nu,y(\nu)) +  g(\nu)\FL[y][\nu] + \Lhat{\nu}} d\nu, 
        \end{aligned}
        \label{eqn:L1DRAC:Definition:FeedbackOperator:Predictor}
        \quad \text{\emph{(Process Predictor)}}
    \end{align}
\end{subequations}
for $t \in [0,T]$, where $\Boldomega, \BoldTs, \lambda_s \in \mathbb{R}_{>0}$ are the control parameters.
The parameters $\Boldomega$ and $\BoldTs$ and are referred to as the \textbf{\emph{filter bandwidth}} and the \textbf{\emph{sampling period}}, respectively.  
Additionally, $ \Theta_{ad}(t) = \begin{bmatrix}\mathbb{I}_m & 0_{m,n-m}  \end{bmatrix} \bar{g}(t)^{-1} \in \mathbb{R}^{m \times n}$, where $\bar{g}(t) = \begin{bmatrix} g(t) & g(t)^\perp  \end{bmatrix} \in \mathbb{R}^{n \times n}$, and here $g^\perp$ is defined in Assumption~\ref{assmp:VectorFields}.
Finally, $\bar{\xi} \sim \bar{\mathcal{D}}$ in~\eqref{eqn:L1DRAC:Definition:FeedbackOperator:Predictor} is the initial condition of the true (uncertain) process~\eqref{eqn:L1DRACProcess}.

The choice of the control parameters $\Boldomega \in \mathbb{R}_{>0}$ and $\BoldTs \in \mathbb{R}_{>0}$, the bandwidth for the low-pass filter in~\eqref{eqn:L1DRAC:Definition:FeedbackOperator:Filter} and the sampling period for the adaptation law in~\eqref{eqn:L1DRAC:Definition:FeedbackOperator:AdaptationLaw}, respectively depend on the growth rate of the uncertain vector fields in Assumption~\ref{assmp:VectorFields}, and the $\delta$-ISS property of the nominal system in Assumption~\ref{assmp:ILF}.
We omit the details of the control parameter selection, and instead refer the reader to~\cite[Sec.~3.2]{L1DRAC} for a comprehensive exposition. 
The guarantees of the \ellonedrac~controller are summarized in the following theorem.
\begin{theorem}[\cite{L1DRAC}]\label{thm:L1DRAC}
    Fix a failure probability $\delta \in (0,1)$ and let $\mathbb{N}_{\geq 1} \ni p \geq \log\br{\sqrt{1/\delta}}$.
    Suppose that $\max_{t \in [0,T]}\mathbb{E}_{\xi \sim \mathcal{D}}\left[\norm{\xnomt{t}{\xi; \fdbktasil }}^{2p}\right]^\frac{1}{2p}\leq \Delta_\star$, then there exist an a-priori known constant $\rho_{\mathcal{L}_1}(p) = \rho_{\mathcal{L}_1}\br{p,\delta,\Delta_\star, \Delta_\mu, \Delta_\sigma} \in \mathbb{R}_{>0}$ such that 
    \begin{align}\label{eqn:result:1}
        \widetilde{\Upsilon}_T \left(\widetilde{\mathcal{D}}; \pi_{ad}, \fdbktasil  \right)
        \leq 
        \rho_{\mathcal{L}_1}(p),
    \end{align}
    where $\pi_{ad} = \fdbktasil + \fdbkellone$.
    Moreover, with probability $1-\delta$
    \begin{align}\label{eqn:result:2}
        \norm{\Xt{t}\left(\bar{\xi}; \pi_{ad} \right) - \xnomt{t}{\xi; \fdbktasil}} < 
        e \cdot \rho_{\mathcal{L}_1}\left(\log \sqrt{ \frac{1}{\delta} } \right).
    \end{align}

\end{theorem}
A few remarks are in order. 
The control law $\pi_{ad} = \fdbktasil + \fdbkellone$ provides the following uniform bound on
\begin{align*}
    \max_{t \in [0,T]}
        \mathbb{E}_{\widetilde{\xi} \sim \widetilde{\mathcal{D}}}
        \left[\norm{\Xt{t}\left(\bar{\xi}; \pi_{ad} \right) - \xnomt{t}{\xi; \fdbktasil}}^{2p}\right]^\frac{1}{2p}
        \leq \rho_{\mathcal{L}_1}(p),
\end{align*}
from which~\eqref{eqn:result:1} follows by the inclusion of $L_p$ spaces. 
Consequently,~\eqref{eqn:result:2} follows then by the Markov inequality. 

An important requirement in Theorem~\ref{thm:L1DRAC} is the uniform boundedness of $\mathbb{E}_{\xi \sim \mathcal{D}}\left[\norm{\xnomt{t}{\xi; \fdbktasil }}^{2p}\right]^\frac{1}{2p}\leq \Delta_\star$, for a desired $p$.
While TaSIL, as set up does not directly provide this, changing the loss function of TaSIL for training the imitation feedback policy $\fdbktasil$ is a direction that we wish to explore in future work.
Further conditions under which the requirement may be satisfied is if the nominal system under the TaSIL policy, remains bounded over the horizon $[0,T]$. 
For the current setup, if we define the total imitation gap to be between the \textbf{dataset} of expert trajectories $\widehat{D}_n$, instead of $\mathcal{D}$, then the requirement is satisfied owing the compact support of $\mathcal{D}$.
Note that in this case, there is no restriction on the use of $\pi_{ad}$ policy, that still remains defined as is and initialized with the same arbitrary distribution $\mathcal{D}$, which is not restricted to be compactly supported.  
Finally, the \elloneDRAC control law requires the contraction of the baseline system, which in our case would be the nominal (TaSIL) system.  
However, Assumption~\ref{assmp:ILF} only provides a contraction property for the nominal system under the expert policy $\pi^\star$. 
Fortunately, the assumed regularity of the policies in Assumption~\ref{assmp:PolicyRegularity} allows us to follow the same approach as for TaSIL in which the \elloneDRAC will be designed to be robust against the additional perturbation of linearly growing upper bound on the policy shift.

%%%%%%%%%%%%%%%%%%%%%%%%%%%%%%%%%%%%%%%%%%%%%%%%%%%%%%%%%%%%
\subsection{Discussion}\label{sec:discussion}
%%%%%%%%%%%%%%%%%%%%%%%%%%%%%%%%%%%%%%%%%%%%%%%%%%%%%%%%%%%%

Substituting the results from Theorems~\ref{thm:TaSIL} and~\ref{thm:L1DRAC} into~\eqref{eqn:Decomposition:IG:Final} produces the value of the total imitation gap in~\eqref{eqn:TIG} as
\begin{align}\label{eqn:TotalTIG}
    \Upsilon_T \left(\widetilde{\mathcal{D}}; \pi_{ad} = \fdbktasil + \fdbkellone  \right)
    \leq \rho = \rho_{\text{\tiny TaSIL}} + \rho_{\mathcal{L}_1}.
\end{align}
As discussed above, this result holds when $\widetilde{\mathcal{D}}$ is a coupling between $\widehat{\mathcal{D}}_n$ and $\bar{\mathcal{D}}$, and not $\mathcal{D}$ and $\bar{\mathcal{D}}$. 
Therefore, our results directly quantify the difference on the performance of the uncertain system to the training distribution $\widehat{D}_n$.
The bound $\rho_{\text{\tiny TaSIL}}(\delta,n)$ in Theorem~\ref{thm:TaSIL} belongs to $ \mathcal{O}\left(\frac{\log n}{n}\right)$, as $n \uparrow \infty$ which shows a `fast' decrease in the policy-IG as a function of the available training data.  
Moreover, the TaSIL bound also belongs to $\mathcal{O}\left(\frac{1}{\delta}\right)$, as  $\delta \downarrow 0$ establishing a direct relationship.
The \ellonedrac bound $\rho_{\mathcal{L}_1}$ is $\mathcal{O}(1)$ in $n$ since it is training-free, and displays a behavior of $\mathcal{O}\left(\log^2 \sqrt{ \frac{1}{\delta} } \right)$ as $\delta \downarrow 0$, see~\cite[Sec.~4.3]{L1DRAC}.
Thus the \ellonedrac bound increases much slower than the TaSIL bound as $\delta \downarrow 0$.
Hence, the sample complexity, and the complexity of success probability of the total TaSIL and \ellonedrac is identical to TaSIL only.   
It is important to note that the bounds of \elloneDRAC and TaSIL can be consolidated at the instantaneous distributional level due to the nature of guarantees provided by \elloneDRAC.
In the absence of uncertainties, TaSIL guarantees bounds on the imitation gap on the pathwise/trajectory level.

%%%%%%%%%%%%%%%%%%%%%%%%%%%%%%%%%%%%%%%%%%%%%%%%%%%%%%%%%%%%%%%%%%%%%%%%%%%%%%%%
\section{Numerical Experiments}\label{sec:numerical}
%%%%%%%%%%%%%%%%%%%%%%%%%%%%%%%%%%%%%%%%%%%%%%%%%%%%%%%%%%%%%%%%%%%%%%%%%%%%%%%%
In this section, we validate our proposed approach by simulating the system with TaSIL along with \elloneDRAC on an uncertain system~\eqref{eqn:L1DRACProcess} described by:
\begin{equation*}
    f(t,X) = -0.05\mathbb{I}_{4}X_{t} + 0.25\mathbb{I}_{4}(\pi_{\text{\tiny TaSIL}} - h(X_{t})) \in \mathbb{R}^4,
\end{equation*}
where $h(X_{t}) \in \mathbb{R}^{n}$ is a known function which we set as a neural network.  The input operator is defined by:
\begin{equation*}
    g(t) = 0.25\mathbb{I}_{4} \in \mathbb{S}^4. 
\end{equation*}
The drift and diffusion uncertainties are chosen as
\begin{equation*}
    \Lambda_{\mu}(t,X) = \mathbb{I}_{4}(0.1+0.05\norm{X}), \quad \Lambda_{\sigma}(t,X) = \mathbb{I}_{4}(0.1 + 0.05\norm{X}^{0.5}).
\end{equation*}
The expert policy is set to $-Kx - h(x)$, where the feedback $K$ stabilizes the known system by removing $h(x)$. 
The training data set consists of $20$ trajectories under this feedback on the nominal system.
We experimentally demonstrate i) the effect of imitation gap between TaSIL and expert in the absence of uncertainty-induced distribution shifts, ii) the imitation gap with only TaSIL for the uncertain system, and iii) the imitation gap for the uncertain system with both TaSIL and \ellonedrac. 
The results are illustrated in~Fig.~\ref{fig:placeholder} for an ensemble of $100$ trajectories.

\begin{figure}
    \centering
    \includegraphics[width=0.6\linewidth]{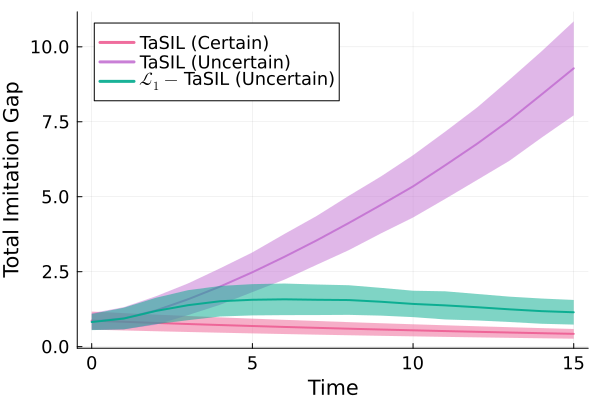}
    \caption{Comparison of total imitation gap under i) TaSIL for the nominal system, ii) TaSIL for the uncertain system, and iii) TaSIL and \elloneDRAC for the uncertain system.}
    \label{fig:placeholder}
\end{figure}

While TaSIL on the nominal system behaves as predicted, the presence of uncertainties lead to the destabilization of the uncertain system. The joint law of TaSIL and \ellonedrac stabilize the system, and keep the imitation gap bounded by delegating the sources of imitation gap in a decoupled fashion.

%%%%%%%%%%%%%%%%%%%%%%%%%%%%%%%%%%%%%%%%%%%%%%%%%%%%%%%%%%%%%%%%%%%%%%%%%%%%%%%%
\section{Conclusion}\label{sec:conclusion}
%%%%%%%%%%%%%%%%%%%%%%%%%%%%%%%%%%%%%%%%%%%%%%%%%%%%%%%%%%%%%%%%%%%%%%%%%%%%%%%%

We provide an exploratory foray into consolidating learning-base imitation learning with robust adaptive control. 
The enabling feature for the consolidation lies with the fact that both the imitation learning approach of TaSIL~\cite{pfrommer2022tasil} and \ellonedrac~\cite{L1DRAC} provides robustness at the levels of distributions thus allowing the relation to the training data itself. 

This approach opens up a new direction of providing guarantees of learning-based controllers without further requiring limiting assumptions on either the learning algorithms or the control methodologies. 
The next step for the research is to explicitly establish the training data-dependent guarantees on the methodology presented in the manuscript by enabling a bi-directional communication between the two components, which can then be communicated to a high-level planner. 
The distributional guarantees also enable the certifiable use of vision-based perception based on the training data. 

%%%%%%%%%%%%%%%%%%%%%%%%%%%%%%%%%%%%%%%%%%%%%%%%%%%%%%%%%%%%%%%%%%%%%%%%%%%%%%%%
\section*{Acknowledgments}
%%%%%%%%%%%%%%%%%%%%%%%%%%%%%%%%%%%%%%%%%%%%%%%%%%%%%%%%%%%%%%%%%%%%%%%%%%%%%%%%
This work is supported by the Air Force Office of Scientific Research Grant (AFOSR) Grant FA9550-21-1-0411, the National Aeronautics and Space Administration (NASA) under Grants 80NSSC22M0070 and 80NSSC20M0229, and by the National Science Foundation (NSF) under Grants CMMI 2135925 and IIS 2331878.

\bibliography{References-Bib/References, References-Bib/Adi}

@inproceedings{lakshmanan2020safe,
  title={Safe feedback motion planning: {A} contraction theory and $\mathcal{L}_1$-adaptive control based approach},
  author={Lakshmanan, Arun and Gahlawat, Aditya and Hovakimyan, Naira},
  booktitle={2020 59th IEEE Conference on Decision and Control (CDC)},
  pages={1578--1583},
  year={2020},
  organization={IEEE}
}

@inproceedings{gahlawat2021contraction,
  title={Contraction $\mathcal{L}_1$-Adaptive Control using {G}aussian Processes},
  author={Gahlawat, Aditya and Lakshmanan, Arun and Song, Lin and Patterson, Andrew and Wu, Zhuohuan and Hovakimyan, Naira and Theodorou, Evangelos A},
  booktitle={Learning for Dynamics and Control},
  pages={1027--1040},
  year={2021},
  organization={PMLR}
}

@inproceedings{sungrobust,
  title={Robust Model Based Reinforcement Learning Using $\mathcal{L}_1$ Adaptive Control},
  author={Sung, Minjun and Karumanchi, Sambhu Harimanas and Gahlawat, Aditya and Hovakimyan, Naira},
  booktitle={The Twelfth International Conference on Learning Representations},
  year = {2024}
}

@book{oksendal2013stochastic,
  title={Stochastic differential equations: {A}n introduction with applications},
  author={Oksendal, Bernt},
  year={2013},
  publisher={Springer Science \& Business Media}
}

@article{pfrommer2022tasil,
  title={Tasil: {T}aylor series imitation learning},
  author={Pfrommer, Daniel and Zhang, Thomas and Tu, Stephen and Matni, Nikolai},
  journal={Advances in Neural Information Processing Systems},
  volume={35},
  pages={20162--20174},
  year={2022}
}

@book{steyer2017probability,
  title={Probability and conditional expectation: Fundamentals for the empirical sciences},
  author={Steyer, Rolf and Nagel, Werner},
  year={2017},
  publisher={John Wiley \& Sons}
}

@book{villani2009optimal,
  title={Optimal transport: {O}ld and new},
  author={Villani, C{\'e}dric},
  volume={338},
  year={2009},
  publisher={Springer}
}

@book{bulloContractionBook,
  author =    {Bullo, Francesco},
  title =     {Contraction Theory for Dynamical Systems},
  year =      2024,
  edition =   {{1.2}},
  publisher = {Kindle Direct Publishing},
  url =       {https://fbullo.github.io/ctds}
}

@article{DavydovJafarpourBullo2022,
  author    = {Alexander Davydov and Saber Jafarpour and Francesco Bullo},
  title     = {Non-{E}uclidean contraction theory for robust nonlinear stability},
  journal   = {IEEE Transactions on Automatic Control},
  volume    = {67},
  number    = {12},
  pages     = {6667--6681},
  year      = {2022},
  doi       = {10.1109/TAC.2022.3151811},
  publisher = {IEEE}
}

@book{Khalil2002NonlinearSystems,
  author    = {Hassan K. Khalil},
  title     = {Nonlinear Systems},
  edition   = {3rd},
  publisher = {Prentice Hall},
  address   = {Upper Saddle River, NJ},
  year      = {2002}
}

@article{TsukamotoChungSlotine2021,
  author  = {Hiroyasu Tsukamoto and Soon-Jo Chung and Jean-Jacques Slotine},
  title   = {Contraction Theory for Nonlinear Stability Analysis and Learning-based Control: {A} Tutorial Overview},
  journal = {Annual Reviews in Control},
  volume  = {52},
  pages   = {135--169},
  year    = {2021},
  publisher = {Elsevier}
}

@article{angeli2002lyapunov,
  title={A {L}yapunov approach to incremental stability properties},
  author={Angeli, David},
  journal={IEEE Transactions on Automatic Control},
  volume={47},
  number={3},
  pages={410--421},
  year={2002},
  publisher={IEEE}
}

@article{mehtaStROLStabilizedRobust2024,
  title = {{{StROL}}: {{Stabilized}} and {{Robust Online Learning From Humans}}},
  shorttitle = {{{StROL}}},
  author = {Mehta, Shaunak A. and Meng, Forrest and Bajcsy, Andrea and Losey, Dylan P.},
  year = 2024,
  month = mar,
  journal = {IEEE Robotics and Automation Letters},
  volume = {9},
  number = {3},
  pages = {2303--2310},
  issn = {2377-3766},
  doi = {10.1109/LRA.2024.3354626},
  urldate = {2025-12-16}
}

@inproceedings{kangLyapunovDensityModels2022,
  title = {Lyapunov {{Density Models}}: {{Constraining Distribution Shift}} in {{Learning-Based Control}}},
  shorttitle = {Lyapunov {{Density Models}}},
  booktitle = {Proceedings of the 39th {{International Conference}} on {{Machine Learning}}},
  author = {Kang, Katie and Gradu, Paula and Choi, Jason J. and Janner, Michael and Tomlin, Claire and Levine, Sergey},
  year = 2022,
  month = jun,
  urldate = {2025-12-16}
}

@article{pengDeepMimicExampleguidedDeep2018,
  title = {{{DeepMimic}}: Example-Guided Deep Reinforcement Learning of Physics-Based Character Skills},
  shorttitle = {{{DeepMimic}}},
  author = {Peng, Xue Bin and Abbeel, Pieter and Levine, Sergey and {van de Panne}, Michiel},
  year = 2018,
  month = jul,
  journal = {ACM Trans. Graph.},
  volume = {37},
  number = {4},
  pages = {143:1--143:14},
  issn = {0730-0301},
  doi = {10.1145/3197517.3201311},
  urldate = {2025-12-16}
}

@inproceedings{chaeRobustImitationLearning2022,
  title = {Robust {{Imitation Learning}} against {{Variations}} in {{Environment Dynamics}}},
  booktitle = {Proceedings of the 39th {{International Conference}} on {{Machine Learning}}},
  author = {Chae, Jongseong and Han, Seungyul and Jung, Whiyoung and Cho, Myungsik and Choi, Sungho and Sung, Youngchul},
  year = 2022,
  month = jun,
  urldate = {2025-12-16}
}

@inproceedings{wangRobustImitationDiverse2017,
  title = {Robust {{Imitation}} of {{Diverse Behaviors}}},
  booktitle = {Advances in {{Neural Information Processing Systems}}},
  author = {Wang, Ziyu and Merel, Josh S and Reed, Scott E and {de Freitas}, Nando and Wayne, Gregory and Heess, Nicolas},
  year = 2017,
  volume = {30},
  publisher = {Curran Associates, Inc.},
  urldate = {2025-12-16}
}

@article{caiProbabilisticEndtoEndVehicle2020,
  title = {Probabilistic {{End-to-End Vehicle Navigation}} in {{Complex Dynamic Environments With Multimodal Sensor Fusion}}},
  author = {Cai, Peide and Wang, Sukai and Sun, Yuxiang and Liu, Ming},
  year = 2020,
  month = jul,
  journal = {IEEE Robotics and Automation Letters},
  volume = {5},
  number = {3},
  pages = {4218--4224},
  issn = {2377-3766},
  doi = {10.1109/LRA.2020.2994027},
  urldate = {2025-12-16}
}

@misc{sermanetTimeContrastiveNetworksSelfSupervised2018,
  title = {Time-{{Contrastive Networks}}: {{Self-Supervised Learning}} from {{Video}}},
  shorttitle = {Time-{{Contrastive Networks}}},
  author = {Sermanet, Pierre and Lynch, Corey and Chebotar, Yevgen and Hsu, Jasmine and Jang, Eric and Schaal, Stefan and Levine, Sergey},
  year = 2018,
  month = mar,
  number = {arXiv:1704.06888},
  eprint = {1704.06888},
  primaryclass = {cs},
  publisher = {arXiv},
  doi = {10.48550/arXiv.1704.06888},
  urldate = {2025-12-16}
}

@inproceedings{liuImitationObservationLearning2018,
  title = {Imitation from {{Observation}}: {{Learning}} to {{Imitate Behaviors}} from {{Raw Video}} via {{Context Translation}}},
  shorttitle = {Imitation from {{Observation}}},
  booktitle = {2018 {{IEEE International Conference}} on {{Robotics}} and {{Automation}} ({{ICRA}})},
  author = {Liu, YuXuan and Gupta, Abhishek and Abbeel, Pieter and Levine, Sergey},
  year = 2018,
  month = may,
  pages = {1118--1125},
  issn = {2577-087X},
  doi = {10.1109/ICRA.2018.8462901},
  urldate = {2025-12-16}
}

@misc{stadieThirdPersonImitationLearning2019,
  title = {Third-{{Person Imitation Learning}}},
  author = {Stadie, Bradly C. and Abbeel, Pieter and Sutskever, Ilya},
  year = 2019,
  month = sep,
  number = {arXiv:1703.01703},
  eprint = {1703.01703},
  primaryclass = {cs},
  publisher = {arXiv},
  doi = {10.48550/arXiv.1703.01703},
  urldate = {2025-12-16}
}

@misc{caiGlobalConvergenceImitation2019,
  title = {On the {{Global Convergence}} of {{Imitation Learning}}: {{A Case}} for {{Linear Quadratic Regulator}}},
  shorttitle = {On the {{Global Convergence}} of {{Imitation Learning}}},
  author = {Cai, Qi and Hong, Mingyi and Chen, Yongxin and Wang, Zhaoran},
  year = 2019,
  month = jan,
  number = {arXiv:1901.03674},
  eprint = {1901.03674},
  primaryclass = {cs},
  publisher = {arXiv},
  doi = {10.48550/arXiv.1901.03674},
  urldate = {2025-12-16}
}

@misc{leeSafeEndtoendImitation2019,
  title = {Safe End-to-End Imitation Learning for Model Predictive Control},
  author = {Lee, Keuntaek and Saigol, Kamil and Theodorou, Evangelos A.},
  year = 2019,
  month = feb,
  number = {arXiv:1803.10231},
  eprint = {1803.10231},
  primaryclass = {cs},
  publisher = {arXiv},
  doi = {10.48550/arXiv.1803.10231},
  urldate = {2025-12-16}
}

@inproceedings{beliaevImitationLearningEstimating2022,
  title = {Imitation {{Learning}} by {{Estimating Expertise}} of {{Demonstrators}}},
  booktitle = {Proceedings of the 39th {{International Conference}} on {{Machine Learning}}},
  author = {Beliaev, Mark and Shih, Andy and Ermon, Stefano and Sadigh, Dorsa and Pedarsani, Ramtin},
  year = 2022,
  month = jun,
  urldate = {2025-12-16}
}

@inproceedings{
sasaki2021behavioral,
title={{Behavioral Cloning from Noisy Demonstrations}},
author={Fumihiro Sasaki and Ryota Yamashina},
booktitle={International Conference on Learning Representations},
year={2021},
url={https://openreview.net/forum?id=zrT3HcsWSAt}
}

@inproceedings{wuImitationLearningImperfect2019,
  title = {Imitation {{Learning}} from {{Imperfect Demonstration}}},
  booktitle = {Proceedings of the 36th {{International Conference}} on {{Machine Learning}}},
  author = {Wu, Yueh-Hua and Charoenphakdee, Nontawat and Bao, Han and Tangkaratt, Voot and Sugiyama, Masashi},
  year = 2019,
  month = may,
  urldate = {2025-12-16}
}

@article{grollmanRobotLearningFailed2012,
  title = {Robot {{Learning}} from {{Failed Demonstrations}}},
  author = {Grollman, Daniel H. and Billard, Aude G.},
  year = 2012,
  month = jun,
  journal = {International Journal of Social Robotics},
  volume = {4},
  number = {4},
  pages = {331--342},
  publisher = {Springer},
  issn = {1875-4805},
  doi = {10.1007/s12369-012-0161-z},
  urldate = {2025-12-16}
}

@misc{shiWaypointBasedImitationLearning2023,
  title = {Waypoint-{{Based Imitation Learning}} for {{Robotic Manipulation}}},
  author = {Shi, Lucy Xiaoyang and Sharma, Archit and Zhao, Tony Z. and Finn, Chelsea},
  year = 2023,
  month = jul,
  number = {arXiv:2307.14326},
  eprint = {2307.14326},
  primaryclass = {cs},
  publisher = {arXiv},
  doi = {10.48550/arXiv.2307.14326},
  urldate = {2025-12-16}
}

@inproceedings{leSmoothImitationLearning2016,
  title = {Smooth {{Imitation Learning}} for {{Online Sequence Prediction}}},
  booktitle = {Proceedings of {{The}} 33rd {{International Conference}} on {{Machine Learning}}},
  author = {Le, Hoang and Kang, Andrew and Yue, Yisong and Carr, Peter},
  year = 2016,
  month = jun,
  urldate = {2025-12-16}
}

@inproceedings{heImitationLearningCoaching2012,
  title = {Imitation {{Learning}} by {{Coaching}}},
  booktitle = {Advances in {{Neural Information Processing Systems}}},
  author = {He, He and Eisner, Jason and Daume, Hal},
  year = 2012,
  volume = {25},
  publisher = {Curran Associates, Inc.},
  urldate = {2025-12-16}
}

@misc{judahActiveImitationLearning2012,
  title = {Active {{Imitation Learning}} via {{Reduction}} to {{I}}.{{I}}.{{D}}. {{Active Learning}}},
  author = {Judah, Kshitij and Fern, Alan and Dietterich, Thomas G.},
  year = 2012,
  month = oct,
  number = {arXiv:1210.4876},
  eprint = {1210.4876},
  primaryclass = {cs},
  publisher = {arXiv},
  doi = {10.48550/arXiv.1210.4876},
  urldate = {2025-12-16}
}

@misc{reddySQILImitationLearning2019,
  title = {{{SQIL}}: {{Imitation Learning}} via {{Reinforcement Learning}} with {{Sparse Rewards}}},
  shorttitle = {{{SQIL}}},
  author = {Reddy, Siddharth and Dragan, Anca D. and Levine, Sergey},
  year = 2019,
  month = sep,
  number = {arXiv:1905.11108},
  eprint = {1905.11108},
  primaryclass = {cs},
  publisher = {arXiv},
  doi = {10.48550/arXiv.1905.11108},
  urldate = {2025-12-16}
}

@inproceedings{
Brantley2020Disagreement-Regularized,
title={Disagreement-Regularized Imitation Learning},
author={Kiante Brantley and Wen Sun and Mikael Henaff},
booktitle={International Conference on Learning Representations},
year={2020},
url={https://openreview.net/forum?id=rkgbYyHtwB}
}

@inproceedings{wangRandomExpertDistillation2019,
  title = {Random {{Expert Distillation}}: {{Imitation Learning}} via {{Expert Policy Support Estimation}}},
  shorttitle = {Random {{Expert Distillation}}},
  booktitle = {Proceedings of the 36th {{International Conference}} on {{Machine Learning}}},
  author = {Wang, Ruohan and Ciliberto, Carlo and Amadori, Pierluigi Vito and Demiris, Yiannis},
  year = 2019,
  month = may,
  pages = {6536--6544},
  publisher = {PMLR},
  issn = {2640-3498},
  urldate = {2025-12-16},
  langid = {english}
}

@misc{zhangQueryEfficientImitationLearning2016,
  title = {Query-{{Efficient Imitation Learning}} for {{End-to-End Autonomous Driving}}},
  author = {Zhang, Jiakai and Cho, Kyunghyun},
  year = 2016,
  month = may,
  number = {arXiv:1605.06450},
  eprint = {1605.06450},
  primaryclass = {cs},
  publisher = {arXiv},
  doi = {10.48550/arXiv.1605.06450},
  urldate = {2025-12-16},
  archiveprefix = {arXiv}
}

@misc{hoqueThriftyDAggerBudgetAwareNovelty2021,
  title = {{{ThriftyDAgger}}: {{Budget-Aware Novelty}} and {{Risk Gating}} for {{Interactive Imitation Learning}}},
  shorttitle = {{{ThriftyDAgger}}},
  author = {Hoque, Ryan and Balakrishna, Ashwin and Novoseller, Ellen and Wilcox, Albert and Brown, Daniel S. and Goldberg, Ken},
  year = 2021,
  month = sep,
  number = {arXiv:2109.08273},
  eprint = {2109.08273},
  primaryclass = {cs},
  publisher = {arXiv},
  doi = {10.48550/arXiv.2109.08273},
  urldate = {2025-12-16},
  archiveprefix = {arXiv}
}

@inproceedings{kellyHGDAggerInteractiveImitation2019,
  title = {{{HG-DAgger}}: {{Interactive Imitation Learning}} with {{Human Experts}}},
  shorttitle = {{{HG-DAgger}}},
  booktitle = {2019 {{International Conference}} on {{Robotics}} and {{Automation}} ({{ICRA}})},
  author = {Kelly, Michael and Sidrane, Chelsea and {Driggs-Campbell}, Katherine and Kochenderfer, Mykel J.},
  year = 2019,
  month = may,
  pages = {8077--8083},
  issn = {2577-087X},
  doi = {10.1109/ICRA.2019.8793698},
  urldate = {2025-12-16}
}

@article{lemeroSurveyImitationLearning2022,
  title = {A {{Survey}} on {{Imitation Learning Techniques}} for {{End-to-End Autonomous Vehicles}}},
  author = {Le Mero, Luc and Yi, Dewei and Dianati, Mehrdad and Mouzakitis, Alexandros},
  year = 2022,
  month = sep,
  journal = {IEEE Transactions on Intelligent Transportation Systems},
  volume = {23},
  number = {9},
  pages = {14128--14147},
  issn = {1558-0016},
  doi = {10.1109/TITS.2022.3144867},
  urldate = {2025-12-15}
}

@article{ravichandarRecentAdvancesRobot2020,
  title = {Recent {{Advances}} in {{Robot Learning}} from {{Demonstration}}},
  author = {Ravichandar, Harish and Polydoros, Athanasios S. and Chernova, Sonia and Billard, Aude},
  year = 2020,
  month = may,
  publisher = {Annual Reviews},
  doi = {10.1146/annurev-control-100819-063206},
  urldate = {2025-12-15},
  langid = {english}
}

@article{ROB-053,
url = {http://dx.doi.org/10.1561/2300000053},
year = {2018},
volume = {7},
journal = {Foundations and Trends® in Robotics},
title = {An Algorithmic Perspective on Imitation Learning},
doi = {10.1561/2300000053},
issn = {1935-8253},
number = {1-2},
pages = {1-179},
author = {Takayuki Osa and Joni Pajarinen and Gerhard Neumann and J. Andrew Bagnell and Pieter Abbeel and Jan Peters}
}

@article{argallSurveyRobotLearning2009,
  title = {A Survey of Robot Learning from Demonstration},
  author = {Argall, Brenna D. and Chernova, Sonia and Veloso, Manuela and Browning, Brett},
  year = 2009,
  month = may,
  journal = {Robotics and Autonomous Systems},
  volume = {57},
  number = {5},
  pages = {469--483},
  issn = {09218890},
  doi = {10.1016/j.robot.2008.10.024},
  urldate = {2025-12-15},
  copyright = {https://www.elsevier.com/tdm/userlicense/1.0/},
  langid = {english}
}

@book{calinonRobotProgrammingDemonstration2009,
  title = {Robot {{Programming}} by {{Demonstration}}},
  author = {Calinon, Sylvain},
  year = 2009,
  month = aug,
  publisher = {EPFL Press},
  isbn = {978-1-4398-0867-2},
  langid = {english}
}

@article{schaalImitationLearningRoute,
  title = {Is Imitation Learning the Route to Humanoid Robots?},
  author = {Schaal, Stefan},
  urldate = {2025-12-15},
  langid = {english}
}

@inproceedings{abbeelApprenticeshipLearningInverse2004c,
  title = {Apprenticeship Learning via Inverse Reinforcement Learning},
  booktitle = {Proceedings of the Twenty-First International Conference on {{Machine}} Learning},
  author = {Abbeel, Pieter and Ng, Andrew Y.},
  year = 2004,
  month = jul,
  series = {{{ICML}} '04},
  pages = {1},
  publisher = {Association for Computing Machinery},
  address = {New York, NY, USA},
  doi = {10.1145/1015330.1015430},
  urldate = {2025-12-15},
  isbn = {978-1-58113-838-2}
}

@misc{panAgileAutonomousDriving2019,
  title = {Agile {{Autonomous Driving}} Using {{End-to-End Deep Imitation Learning}}},
  author = {Pan, Yunpeng and Cheng, Ching-An and Saigol, Kamil and Lee, Keuntaek and Yan, Xinyan and Theodorou, Evangelos and Boots, Byron},
  year = 2019,
  month = aug,
  number = {arXiv:1709.07174},
  eprint = {1709.07174},
  primaryclass = {cs},
  publisher = {arXiv},
  doi = {10.48550/arXiv.1709.07174},
  urldate = {2025-12-15},
  archiveprefix = {arXiv}
}

@article{yinImitationLearningStability2022,
  title = {Imitation {{Learning With Stability}} and {{Safety Guarantees}}},
  author = {Yin, He and Seiler, Peter and Jin, Ming and Arcak, Murat},
  year = 2022,
  journal = {IEEE Control Systems Letters},
  volume = {6},
  pages = {409--414},
  issn = {2475-1456},
  doi = {10.1109/LCSYS.2021.3077861},
  urldate = {2025-12-09}
}

@article{hertneckLearningApproximateModel2018,
  title = {Learning an {{Approximate Model Predictive Controller With Guarantees}}},
  author = {Hertneck, Michael and K{\"o}hler, Johannes and Trimpe, Sebastian and Allg{\"o}wer, Frank},
  year = 2018,
  month = jul,
  journal = {IEEE Control Systems Letters},
  volume = {2},
  number = {3},
  pages = {543--548},
  issn = {2475-1456},
  doi = {10.1109/LCSYS.2018.2843682},
  urldate = {2025-12-09}
}

@inproceedings{tuSampleComplexityStability2022,
  title = {On the {{Sample Complexity}} of {{Stability Constrained Imitation Learning}}},
  booktitle = {Proceedings of {{The}} 4th {{Annual Learning}} for {{Dynamics}} and {{Control Conference}}},
  author = {Tu, Stephen and Robey, Alexander and Zhang, Tingnan and Matni, Nikolai},
  year = 2022,
  month = may,
  urldate = {2025-12-09},
  langid = {english}
}

@misc{hoGenerativeAdversarialImitation2016a,
  title = {Generative {{Adversarial Imitation Learning}}},
  author = {Ho, Jonathan and Ermon, Stefano},
  year = 2016,
  month = jun,
  number = {arXiv:1606.03476},
  eprint = {1606.03476},
  primaryclass = {cs},
  publisher = {arXiv},
  doi = {10.48550/arXiv.1606.03476},
  urldate = {2025-12-08},
  archiveprefix = {arXiv}
}

@misc{L1DRAC,
  author={Gahlawat, Aditya and H. Karumanchi, Sambhu and Hovakimyan, Naira},
  title={{$\mathcal{L}_{1}$-DRAC}: Distributionally Robust Adaptive Control},
  howpublished = {\url{https://adityagahlawat.github.io/Preprints/DRAC.pdf}},
  year         = {2024},
}

@book{hovakimyan2010ℒ1,
  title={$\mathcal{L}_1$ Adaptive Control Theory: {G}uaranteed Robustness with Fast Adaptation},
  author={Hovakimyan, Naira and Cao, Chengyu},
  year={2010},
  publisher={SIAM}
}

@article{matniQuantitativeFrameworkLayered2024,
  title = {A {{Quantitative Framework}} for {{Layered Multirate Control}}: {{Toward}} a {{Theory}} of {{Control Architecture}}},
  shorttitle = {A {{Quantitative Framework}} for {{Layered Multirate Control}}},
  author = {Matni, Nikolai and Ames, Aaron D. and Doyle, John C.},
  year = {2024},
  month = jun,
  journal = {IEEE Control Systems Magazine},
  volume = {44},
  number = {3},
  pages = {52--94},
  issn = {1941-000X},
  doi = {10.1109/MCS.2024.3382388},
  urldate = {2025-05-20}
}

@misc{matniTheoryControlArchitecture2024,
  title = {Towards a {{Theory}} of {{Control Architecture}}: {{A}} Quantitative Framework for Layered Multi-Rate Control},
  shorttitle = {Towards a {{Theory}} of {{Control Architecture}}},
  author = {Matni, Nikolai and Ames, Aaron D. and Doyle, John C.},
  year = {2024},
  month = jan,
  number = {arXiv:2401.15185},
  eprint = {2401.15185},
  primaryclass = {math},
  publisher = {arXiv},
  doi = {10.48550/arXiv.2401.15185},
  urldate = {2025-05-20},
  archiveprefix = {arXiv}
}

@inproceedings{matniTheoryDynamicsControl2016,
  title = {A {{Theory}} of Dynamics, Control and Optimization in Layered Architectures},
  booktitle = {2016 {{American Control Conference}} ({{ACC}})},
  author = {Matni, Nikolai and Doyle, John C.},
  year = {2016},
  month = jul,
  pages = {2886--2893},
  issn = {2378-5861},
  doi = {10.1109/ACC.2016.7525357},
  urldate = {2025-05-20}
}

@article{singhRobustFeedbackMotion2023,
  title = {Robust Feedback Motion Planning via Contraction Theory},
  author = {Singh, Sumeet and Landry, Benoit and Majumdar, Anirudha and Slotine, Jean-Jacques and Pavone, Marco},
  year = {2023},
  month = aug,
  journal = {The International Journal of Robotics Research},
  volume = {42},
  number = {9},
  pages = {655--688},
  publisher = {SAGE Publications Ltd STM},
  issn = {0278-3649},
  doi = {10.1177/02783649231186165},
  urldate = {2025-05-20}
}

@inproceedings{herbertFaSTrackModularFramework2017,
  title = {{{FaSTrack}}: {{A}} Modular Framework for Fast and Guaranteed Safe Motion Planning},
  shorttitle = {{{FaSTrack}}},
  booktitle = {2017 {{IEEE}} 56th {{Annual Conference}} on {{Decision}} and {{Control}} ({{CDC}})},
  author = {Herbert, Sylvia L. and Chen, Mo and Han, SooJean and Bansal, Somil and Fisac, Jaime F. and Tomlin, Claire J.},
  year = {2017},
  month = dec,
  pages = {1517--1522},
  doi = {10.1109/CDC.2017.8263867},
  urldate = {2025-05-20}
}

@article{mitchellTimedependentHamiltonJacobiFormulation2005,
  title = {A Time-Dependent {{Hamilton-Jacobi}} Formulation of Reachable Sets for Continuous Dynamic Games},
  author = {Mitchell, I.M. and Bayen, A.M. and Tomlin, C.J.},
  year = {2005},
  month = jul,
  journal = {IEEE Transactions on Automatic Control},
  volume = {50},
  number = {7},
  pages = {947--957},
  issn = {1558-2523},
  doi = {10.1109/TAC.2005.851439},
  urldate = {2025-05-20}
}

@inproceedings{NEURIPS2022_7f10c3d6,
  title = {{{TaSIL}}: {{Taylor}} Series Imitation Learning},
  booktitle = {Advances in Neural Information Processing Systems},
  author = {Pfrommer, Daniel and Zhang, Thomas and Tu, Stephen and Matni, Nikolai},
  editor = {Koyejo, S. and Mohamed, S. and Agarwal, A. and Belgrave, D. and Cho, K. and Oh, A.},
  year = {2022},
  volume = {35},
  pages = {20162--20174},
  publisher = {Curran Associates, Inc.}
}

@inproceedings{laskeyDARTNoiseInjection2017,
  title = {{{DART}}: {{Noise Injection}} for {{Robust Imitation Learning}}},
  shorttitle = {{{DART}}},
  booktitle = {Proceedings of the 1st {{Annual Conference}} on {{Robot Learning}}},
  author = {Laskey, Michael and Lee, Jonathan and Fox, Roy and Dragan, Anca and Goldberg, Ken},
  year = {2017},
  month = oct,
  pages = {143--156},
  publisher = {PMLR},
  issn = {2640-3498},
  urldate = {2025-05-20},
  langid = {english}
}

@inproceedings{sunDualPolicyIteration2018,
  title = {Dual {{Policy Iteration}}},
  booktitle = {Advances in {{Neural Information Processing Systems}}},
  author = {Sun, Wen and Gordon, Geoffrey J and Boots, Byron and Bagnell, J.},
  year = {2018},
  volume = {31},
  publisher = {Curran Associates, Inc.},
  urldate = {2025-05-20}
}

@inproceedings{rossReductionImitationLearning2011a,
  title = {A {{Reduction}} of {{Imitation Learning}} and {{Structured Prediction}} to {{No-Regret Online Learning}}},
  booktitle = {Proceedings of the {{Fourteenth International Conference}} on {{Artificial Intelligence}} and {{Statistics}}},
  author = {Ross, Stephane and Gordon, Geoffrey and Bagnell, Drew},
  year = {2011},
  month = jun,
  pages = {627--635},
  publisher = {{JMLR Workshop and Conference Proceedings}},
  issn = {1938-7228},
  urldate = {2025-05-20},
  langid = {english}
}

@inproceedings{rossReductionImitationLearning2011,
  title = {A {{Reduction}} of {{Imitation Learning}} and {{Structured Prediction}} to {{No-Regret Online Learning}}},
  booktitle = {Proceedings of the {{Fourteenth International Conference}} on {{Artificial Intelligence}} and {{Statistics}}},
  author = {Ross, Stephane and Gordon, Geoffrey and Bagnell, Drew},
  year = {2011},
  month = jun,
  pages = {627--635},
  publisher = {{JMLR Workshop and Conference Proceedings}},
  issn = {1938-7228},
  urldate = {2025-05-19},
  langid = {english}
}

@article{husseinDeepImitationLearning2018,
  title = {Deep Imitation Learning for {{3D}} Navigation Tasks},
  author = {Hussein, Ahmed and Elyan, Eyad and Gaber, Mohamed Medhat and Jayne, Chrisina},
  year = {2018},
  month = apr,
  journal = {Neural Computing and Applications},
  volume = {29},
  number = {7},
  pages = {389--404},
  issn = {0941-0643, 1433-3058},
  doi = {10.1007/s00521-017-3241-z},
  urldate = {2025-05-19},
  langid = {english}
}

@inproceedings{abbeelApprenticeshipLearningInverse2004a,
  title = {Apprenticeship Learning via Inverse Reinforcement Learning},
  booktitle = {Twenty-First International Conference on {{Machine}} Learning  - {{ICML}} '04},
  author = {Abbeel, Pieter and Ng, Andrew Y.},
  year = {2004},
  pages = {1},
  publisher = {ACM Press},
  address = {Banff, Alberta, Canada},
  doi = {10.1145/1015330.1015430},
  urldate = {2025-05-19},
  langid = {english}
}

@inproceedings{codevillaEndtoEndDrivingConditional2018,
  title = {End-to-{{End Driving Via Conditional Imitation Learning}}},
  booktitle = {2018 {{IEEE International Conference}} on {{Robotics}} and {{Automation}} ({{ICRA}})},
  author = {Codevilla, Felipe and M{\"u}ller, Matthias and L{\'o}pez, Antonio and Koltun, Vladlen and Dosovitskiy, Alexey},
  year = {2018},
  month = may,
  pages = {4693--4700},
  issn = {2577-087X},
  doi = {10.1109/ICRA.2018.8460487},
  urldate = {2025-05-19}
}

@inproceedings{pomerleauALVINNAutonomousLand1988,
  title = {{{ALVINN}}: {{An Autonomous Land Vehicle}} in a {{Neural Network}}},
  shorttitle = {{{ALVINN}}},
  booktitle = {Advances in {{Neural Information Processing Systems}}},
  author = {Pomerleau, Dean A.},
  year = {1988},
  volume = {1},
  publisher = {Morgan-Kaufmann},
  urldate = {2025-05-19}
}

@misc{ciftciSAFEGILSAFEtyGuided2024a,
  title = {{{SAFE-GIL}}: {{SAFEty Guided Imitation Learning}} for {{Robotic Systems}}},
  shorttitle = {{{SAFE-GIL}}},
  author = {Ciftci, Yusuf Umut and Chiu, Darren and Feng, Zeyuan and Sukhatme, Gaurav S. and Bansal, Somil},
  year = {2024},
  month = nov,
  number = {arXiv:2404.05249},
  eprint = {2404.05249},
  primaryclass = {cs},
  publisher = {arXiv},
  doi = {10.48550/arXiv.2404.05249},
  urldate = {2025-05-19},
  archiveprefix = {arXiv}
}

@inproceedings{ankileJUICERDataEfficientImitation2024,
  title = {{{JUICER}}: {{Data-Efficient Imitation Learning}} for {{Robotic Assembly}}},
  shorttitle = {{{JUICER}}},
  booktitle = {2024 {{IEEE}}/{{RSJ International Conference}} on {{Intelligent Robots}} and {{Systems}} ({{IROS}})},
  author = {Ankile, Lars and Simeonov, Anthony and Shenfeld, Idan and Agrawal, Pulkit},
  year = {2024},
  month = oct,
  pages = {5096--5103},
  issn = {2153-0866},
  doi = {10.1109/IROS58592.2024.10802498},
  urldate = {2025-05-19}
}

@article{husseinImitationLearningSurvey2018,
  title = {Imitation {{Learning}}: {{A Survey}} of {{Learning Methods}}},
  shorttitle = {Imitation {{Learning}}},
  author = {Hussein, Ahmed and Gaber, Mohamed Medhat and Elyan, Eyad and Jayne, Chrisina},
  year = {2018},
  month = mar,
  journal = {ACM Computing Surveys},
  volume = {50},
  number = {2},
  pages = {1--35},
  issn = {0360-0300, 1557-7341},
  doi = {10.1145/3054912},
  urldate = {2025-05-19},
  langid = {english}
}

\end{document}